\documentclass[11pt]{amsart}

\usepackage{amsmath}
\usepackage{amssymb}
\usepackage[final]{graphicx}

\newtheorem{thm}{Theorem}
\newtheorem{lem}[thm]{Lemma}
\newtheorem{prop}[thm]{Proposition}
\newtheorem{cor}[thm]{Corollary}

\theoremstyle{definition}
\newtheorem{defn}{Definition}[section]

\theoremstyle{remark}
\newtheorem{remark}{Remark}[thm] 

\theoremstyle{plain}

\numberwithin{equation}{section}

\def\CC{{\mathbb C}}

\def\HH{{\mathbb H}}
\def\NN{{\mathbb N}}

\def\RR{{\mathbb R}}

\def\e{\mathrm{e}}
\def\i{\mathrm{i}}

\def\id{\operatorname{id}}

\def\D{\operatorname{D{}}}
\def\E{\operatorname{E{}}}

\def\G{\operatorname{G{}}}

\def\SL{\operatorname{SL}}

\def\PSL{\operatorname{PSL}}

\def\Vol{\operatorname{Vol}}

\def\scrI{{\mathcal I}}

\def\Re{\operatorname{Re}}
\def\Im{\operatorname{Im}}

\newcommand{\be}{\begin{equation}}
\newcommand{\ee}{\end{equation}}
\newcommand{\ba}{\begin{align}}
\newcommand{\ea}{\end{align}}

\begin{document}

\begin{center}
\title[The trace formula for a point scatterer]{The trace formula for a point scatterer \\ on a hyperbolic surface with one cusp}
\author{Henrik Uebersch\"ar}
\address{The Raymond and Beverly Sackler Schoool of Mathematical Sciences, Tel Aviv University, Tel Aviv 69978,
Israel.}
\email{henrik@math.tau.ac.il}
\date{\today. This work is part of the author's PhD thesis completed at the University of Bristol in January 2011. Support from an EPSRC doctoral training grant is gratefully acknowledged. The author is currently supported by a Minerva Fellowship.}
\maketitle
\end{center}

\begin{abstract}
We prove an exact trace formula for the Laplacian with a delta potential - also known as a point scatterer - on a non-compact hyperbolic surface of finite volume with one cusp. Our formula is an analogue of the Selberg trace formula. We express the difference of perturbed and unperturbed trace as a smooth term plus a sum over combinations of diffractive orbits. 
\end{abstract}

\section{Introduction}

Selberg's trace formula is a central tool in the spectral theory of automorphic forms \cite{Hj2,Hj3,Iw} and has important applications in number theory \cite{Hj1}. It also plays a central role in the theory of quantum chaos \cite{Bz} as an exact analogue of Gutzwiller's celebrated trace formula \cite{CRR,G,PU} which links the distribution of the energy levels of a classically chaotic quantum system to the actions of periodic orbits. Selberg proved his trace formula in 1956 for the Laplacian on weakly symmetric Riemannian spaces \cite{S}. 


Among number theorists point scatterers attracted attention in the early 80s when D. Hejhal gave the mathematical explanation to the appearance of Riemann zeros in a numerical experiment of H. Haas \cite{Ha}. In 1977 Haas had produced zeros of the Riemann zeta function when he computed the eigenvalues of the Laplacian on the modular domain. In 1979 Hejhal demystified \cite{Hj4} this phenomenon when he showed that the apparent eigenfunctions which Haas had discovered featured a logarithmic singularity, which was hard to detect numerically, at the point $z_0=e^{\i\pi/3}$ and thus failed to be eigenfunctions of the Laplacian. The full mathematical explanation of these pseudo-eigenfunctions was given by Colin de Verdi\`ere \cite{CdV} in 1983 by revealing them to be eigenfunctions of rank one perturbations of the Laplacian by a delta potential at the conic singularity. 

Furthermore, point scatterers on flat tori are of interest as models of the transition between integrability and chaos in quantum systems. Whereas their classical dynamics is close to integrable, the associated quantum system is not. In \cite{RU} it is shown that the wave functions of a point scatterer on the torus become generically equidistributed in position space. A similar phenomenon occurs for the Dirichlet or Neumann Laplacian on rational polygons \cite{MR}. The exact trace formula for a point scatterer also has interesting applications \cite{RU2} regarding the investigation of its eigenvalue spacing distribution. 

{\bf Acknowledgements:} The author would like to thank Professor Jens Marklof for his guidance and many helpful discussions during the completion of this work. The author is also grateful to Professors Dennis Hejhal and Andreas Str\"ombergsson for discussions about this work during a visit to the University of Uppsala. Several helpful remarks by Yiannis Petridis and Andrew Booker have led to the improvement of this paper. This research was conducted with the support of an EPSRC doctoral training stipend. The author is currently supported by a Minerva Fellowship.

\section{Background and Results}

\subsection{The Laplacian on a hyperbolic surface}
Denote by $\HH$ the hyperbolic halfplane with Riemannian metric $ds^2=y^{-2}(dx^2+dy^2)$ and volume element $d\mu(z)=y^{-2}dx\,dy$. The Laplace-Beltrami operator on $\HH$ is given by $$\Delta=y^2\left(\frac{\partial^2}{\partial x^2}+\frac{\partial^2}{\partial y^2}\right).$$ Let $\PSL(2,\RR)=\SL(2,\RR)/\lbrace\id,-\id\rbrace$.
The group of orientation preserving isometries of $\HH$ is given by the M\"obius transformations
\begin{equation}
	 z\to\frac{az+b}{cz+d},\qquad
	 \begin{pmatrix}
   a & b\\
   c & d
   \end{pmatrix}
   \in\PSL(2,\RR).
\end{equation}

Let $\Gamma\subset\PSL(2,\RR)$ be a discrete subgroup such that the quotient $X=\Gamma\backslash\HH$ has one cusp and finite volume. We may assume that the cusp is at $\infty$. Denote the stabiliser of the cusp by $\Gamma_\infty$. For $s\in\CC$, $\Re s>1$, and $z\in X$ we define the Eisenstein series on $X$ by
\begin{equation}\label{eisen11}
E(z,s)=\sum_{\gamma\in\Gamma_\infty\backslash\Gamma}(\Im(\gamma z))^s.
\end{equation}
The function $E(z,s)$ extends meromorphically in $s$ to all of $\CC$ and satisfies the functional equation
\begin{equation}
E(z,s)=\varphi(s)E(z,1-s)
\end{equation}
where $\varphi(s)$ denotes the scattering coefficient, and in the cusp, as $y=\Im z\to+\infty$ we have the asymptotics $$E(z,s)\sim y^s+\varphi(s)y^{1-s}.$$

We define the automorphic Green's function $G^\Gamma_s(z,w)$ as the resolvent kernel of the Laplacian on $X$ which satisfies
\begin{equation}
(\Delta+s(1-s))G^\Gamma_s(\cdot,w)=\delta_w.
\end{equation}
In particular (cf. \cite{Iw}, p. 105, eq. (7.19)) it satisfies the functional equation
\begin{equation}\label{Greenfe}
G_{1-s}(z,w)=G_{s}(z,w)-\frac{E(z,s)E(w,1-s)}{1-2s}.
\end{equation}

The spectrum of the Laplacian on $L^2(X)$ consists of a discrete part (Maass forms) and a continous part (Eisenstein series). We have three types of eigenfunctions.
\begin{itemize}
\item[(i)] Residual Maass forms; $(\Delta+\lambda_{k})\varphi_{k}=0$, $\lambda_{k}\in[0,\tfrac{1}{4}]$ which arise as residues of Eisenstein series,\\
\item[(ii)] Maass cusp forms; $(\Delta+\lambda_{k})\varphi_{k}=0$, $\lambda_{k}>0$, and in the cusp as $y\to+\infty$, $\varphi_k(x+\i y)=O(e^{-2\pi y})$,\\
\item[(iii)] Eisenstein series; $(\Delta+\tfrac{1}{4}+r^2)E(z,\tfrac{1}{2}+\i r)=0$, $r\geq0$.
\end{itemize}

\subsection{Definition of the point scatterer}

In this section we will merely review the basic set-up. For a more detailed introduction see section 3.3 in \cite{U} or section 2 in \cite{RU}. For additional background reading consult \cite{Al}, Chapter 1, and \cite{Si}, Chapter 10. 

Denote by $\delta_{z_0}$ the Dirac mass at $z_0\in X$. Formally we associate a point scatterer with the operator $\Delta_{\alpha,z_0}=\Delta+\alpha\delta_{z_0}(\delta_{z_0},\cdot)$, where $\alpha\in\RR_*=\RR\setminus\lbrace0\rbrace$ is a coupling constant. Von Neumann's self-adjoint extension theory allows us to realise the formal operator $-\Delta_{\alpha,z_0}$ as a self-adjoint extension (SAE) of the positive symmetric operator $-\Delta|_{D_0}$, the positive Laplacian restricted to functions which vanish at $z_0$, where $$D_{0}=C^\infty_0(X\setminus\lbrace z_0\rbrace)\cap L^2(X).$$

The symmetric operator $-\Delta|_{D_0}$ has deficiency indices $(1,1)$. The deficiency elements are the automorphic Green functions $G^\Gamma_t(\cdot,z_{0})$, $G^\Gamma_{\bar{t}}(\cdot,z_0)$, where $t(1-t)=\i$, $\Re t>\tfrac{1}{2}$. It therefore admits a one-parameter family of self-adjoint extensions which we will denote by $\lbrace-\Delta_{\varphi}\rbrace_{\varphi\in(-\pi,\pi)}$. In fact there is a one-to-one correspondence between the coupling constant $\alpha\in\RR_*$ and the parameter $\varphi=\varphi(\alpha)$ of the corresponding SAE (cf. \cite{Al}, p. 34. Thm. 1.3.2)
\begin{equation}\label{ext}
-\alpha \lim_{z\to z_0}\Im G^\Gamma_t(z,z_0)=\cot\frac{\varphi}{2}.
\end{equation} 

The domain associated with the SAE $-\Delta_\varphi$ is given by
\begin{equation}
D_\varphi=\lbrace f=g+cG^\Gamma_t+ce^{\i\varphi}G^\Gamma_{\bar{t}}\mid (g,c)\in D_0\times\CC \rbrace
\end{equation}
and it acts as $$-\Delta_\varphi f=-\Delta g+c\i G^\Gamma_t-c\i e^{\i\varphi}G^\Gamma_{\bar{t}}.$$

\subsection{The spectrum}

Fix real $\alpha\in\RR_*$ and $z_0\in X$. We are interested in solutions to the equation
\begin{equation}\label{spec}
(\Delta_{\varphi}+\lambda)\psi_\lambda=0.
\end{equation} 
If $\psi_\lambda\in L^2(X)$, the solutions fall into two groups:

\begin{itemize}
\item[(i)] Old eigenfunctions; $\lambda\in\sigma_d(-\Delta)$ and either has multiplicity $\geq2$ or the old eigenfunction vanishes at $z_0$
\item[(ii)] New eigenfunctions; $\lambda\notin\sigma_d(-\Delta)$, the corresponding eigenfunctions are automorphic Green functions, the new eigenvalues are obtained from a spectral equation, $$\frac{1}{\alpha}=\lim_{z\to z_0}\lbrace G_{1/2+\i\rho}(z,z_0)-\Re G_t(z,z_0)\rbrace,\qquad \lambda=\tfrac{1}{4}+\rho^2.$$ We will show in Theorem \ref{prop} that for $\lambda\geq\tfrac{1}{4}$ this spectral equation is equivalent to $$\frac{1}{\alpha}=\lim_{z\to z_0}\lbrace \Re G_{1/2+\i\rho}(z,z_0)-\Re G_t(z,z_0)\rbrace,$$ $$E(z_0,\tfrac{1}{2}+\i r)=0.$$
\end{itemize}
For more information about the discrete spectrum see also section 2.1 in \cite{U}. The generalised solutions to \eqref{spec} are analogues $E^\alpha(z,\tfrac{1}{2}+\i r)$ of the Eisenstein series \eqref{eisen11} with eigenvalue $\lambda=\tfrac{1}{4}+r^2$, $r\in\RR$. They will be derived in section 4.

\subsection{Results}
The spectrum of $-\Delta_{\varphi}$ on $X$ has a continuous part $[\tfrac{1}{4},\infty)$ which corresponds to perturbed Eisenstein series and a discrete part (cf. \cite{CdV}, Thm. 3) consisting of perturbed cusp forms and a finite number of perturbed residual Maass forms. The perturbed Eisenstein series play the analogous role of the scattering solutions in the unperturbed case. They satisfy a functional equation similar to that satisfied by classical Eisenstein series. We give the proof of the following Theorem in section 4.1.
\begin{thm}\label{thm2}
The generalised eigenfunctions of $-\Delta_{\varphi}$ are the perturbations of the Eisenstein series
\begin{equation}\label{eisen}
E^\alpha(z,s)=E(z,s)-\frac{E(z_{0},s)}{S_\alpha(s)}\,G^{\Gamma}_{s}(z,z_{0})
\end{equation}
where $S_\alpha(s)$ denotes the spectral function of $-\Delta_\varphi$ (see section 3). The perturbed Eisenstein series satisfy
\begin{equation}
(\Delta_\varphi+s(1-s))E^\alpha(z,s)=0.
\end{equation}
In the cusp they have the asymptotics
\begin{equation}
E^\alpha(x+\i y,s)=y^{s}+\varphi_\alpha(s)y^{1-s}+O(\e^{-2\pi y}), \qquad y\to\infty,
\end{equation}
where $\varphi_\alpha(s)$ is the scattering coefficient for the perturbed Laplacian $-\Delta_{\varphi}$. They satisfy the funtional equation
\begin{equation}
E^\alpha(z,s)=\varphi_\alpha(s)E^\alpha(z,1-s).
\end{equation}
\end{thm}
The cuspidal part of the discrete spectrum includes possible degenerate eigenvalues (with multiplicity $>1$) of the Laplacian, but may also include new eigenvalues. To simplify notation we will only list the new eigenvalues which have multiplicity one
\begin{equation}
\lambda^{\alpha}_{-M}<\lambda^{\alpha}_{-M+1}<\cdots<\tfrac{1}{4}\leq\lambda^{\alpha}_{0}<\lambda^{\alpha}_{1}<\cdots<\lambda^{\alpha}_{j}<\cdots.
\end{equation}
The small eigenvalues $\lbrace\lambda^{\alpha}_{j}\rbrace_{j=-M}^{-1}$ correspond to eigenfunctions which are either residues of $E^\alpha(z,s)$ or perturbed cusp forms. We will refer to the residues as residual Maass forms. We will sometimes refer to the small eigenvalues as the residual spectrum, although some of them may be associated with cusp forms instead of residues of the Eisenstein series. The remaining eigenvalues $\lbrace\lambda^{\alpha}_{j}\rbrace_{j\geq0}$ correspond to eigenfunctions which decay exponentially in the cusp and have a logarithmic singularity at $z_{0}$. We will refer to such eigenfunctions as pseudo cusp forms and the associated eigenvalues as the pseudo cuspidal spectrum. For a detailed discussion of the perturbed discrete spectrum see section 2.4. 

We now turn to our second main result, the trace formula for surfaces with one cusp, which we will prove in section 5. We interpret the sum over closed geodesics in the trace formula as a sum over diffractive orbits, where we sum over the number of visits to $z_{0}$ and all combinations of orbits connecting $z_{0}$ to itself on $X$.

\begin{defn}
Let $\sigma\geq\tfrac{1}{2}$ and $\delta>0$. We define $H_{\sigma,\,\delta}$ to be the space of functions $h:\CC\to\CC$, s. t.
\begin{enumerate}
\item[(i)] $h(\rho)=h(-\rho)$ ,
\item[(ii)] $h$ is analytic in the strip $\left|\Im{\rho}\right|\leq\sigma$,
\item[(iii)] $|h(\rho)|\ll(1+|\Re\rho|)^{-2-\delta}$ uniformly in the strip $\left|\Im\rho\right|\leq\sigma$.\\
\end{enumerate}
\end{defn}

Let $m_{\Gamma}=|\scrI|$, where
\begin{equation}
\scrI=\lbrace\gamma\in\Gamma\mid\gamma z_{0}=z_{0}\rbrace.
\end{equation}
Let $\psi(s)=\tfrac{1}{2\pi}\Gamma'(s)/\Gamma(s)$, where $\Gamma(s)$ denotes the Gamma function. For $h\in H_{\sigma,\delta}$, $\beta\in\RR$ and $k\in\NN$ we define the following integral transform of $h$
\begin{equation}\label{trans}
g_{\beta,k}(t)=\frac{(-1)^{k}}{2\pi\i k}\int_{-\i\nu-\infty}^{-\i\nu+\infty}\frac{h'(\rho)e^{-\i\rho t}d\rho}{(1+m_{\Gamma}\beta\psi(\tfrac{1}{2}+\i\rho))^{k}},
\end{equation}
where $\nu>v_{\beta}\in(0,\sigma)$ if $1+m_{\Gamma}\beta\psi(\tfrac{1}{2}+\i\rho)$ has a zero $-\i v_{\beta}$ in the interval $(0,-\i\sigma)$, and $\nu=0$ otherwise. In fact there is at most one zero of $1+m_{\Gamma}\beta\psi(\tfrac{1}{2}+\i\rho)$ in the halfplane $\Im\rho<0$ and it lies on the imaginary axis (for a proof see \cite{U}, p. 4). 

The following trace formula is our main result.
\begin{thm}\label{thm3}
Suppose $\Gamma\subset\PSL(2,\RR)$ is a discrete subgroup, such that $X=\Gamma\backslash\HH$ has one cusp and $\Vol(X)<+\infty$. Let $(\alpha,z_{0})\in\RR_*\times X$. Let $\sigma>\tfrac{1}{2}$ satisfy 
\begin{equation}\label{condition}
(\tfrac{1}{2}+\sigma)^{1/2}\log(\tfrac{1}{2}+\sigma)>C(\Gamma,\alpha,z_{0})
\end{equation}
and also choose $\nu$ as above. 
Let
\begin{equation}\label{renorm}
\beta=\beta(\alpha)=
\begin{cases}
\frac{\alpha}{1-c_0\alpha},\quad \alpha\in\RR\setminus\lbrace-\frac{1
}{c_0}\rbrace\\
\\
-\frac{1}{c_0},\quad \alpha=\pm\infty
\end{cases}
\end{equation}
where
\begin{equation}
c_{0}=m_\Gamma\Re\psi(t)+\Re\sum_{\gamma\in\Gamma\setminus\scrI}G_t(z_{0},\gamma z_{0}).
\end{equation}
For any $\delta>0$ and $h\in H_{\sigma,\,\delta}$ we have the identity
\begin{equation}\label{trace}
\begin{split}
&\sum_{j\geq-M}h(\rho^{\alpha}_{j})-\sum_{j\geq-M}h(\rho_{j})\\
=\;&\frac{1}{2\pi}\int_{-\i\nu-\infty}^{-\i\nu+\infty}h(\rho)\frac{m_\Gamma\beta\psi'(\tfrac{1}{2}+\i\rho)}{1+m_\Gamma\beta\psi(\tfrac{1}{2}+\i\rho)}d\rho+\tfrac{1}{2}\delta_\Gamma h(0)\\
&+\frac{1}{4\pi}\int_{-\infty}^{\infty}h(\rho)\left\{\frac{\varphi_\alpha'}{\varphi_\alpha}(\tfrac{1}{2}+\i\rho)-\frac{\varphi'}{\varphi}(\tfrac{1}{2}+\i\rho)\right\}d\rho\\
&+\sum_{k=1}^{\infty}\,\beta^{k}\sum_{\gamma_{1},\cdots,\gamma_{k}\in\Gamma\setminus\scrI}
\int_{l_{\gamma_{1},z_{0}}}^{\infty}\cdots\int_{l_{\gamma_{k},z_{0}}}^{\infty}\frac{g_{\beta,k}(t_{1}+...+t_{k})\prod_{n=1}^{k}dt_{n}}{\prod_{n=1}^{k}\sqrt{\cosh t_{n}-\cosh l_{\gamma_{n},z_{0}}}}
\end{split}
\end{equation}
where
\begin{equation}
\delta_\Gamma=
\begin{cases}
1,\qquad\text{if $\tfrac{1}{4}$ is not an eigenvalue of the Laplacian}\\
\\
0,\qquad\text{otherwise.}
\end{cases}
\end{equation}
\end{thm}

\section{The spectral function}

We begin with the definition of the spectral function which contains information about new eigenvalues and resonances of the operator $-\Delta_\varphi$ in form of its zeros, as well as information about the eigenvalues of the Laplacian in form of its poles.
\begin{defn}
Let $\alpha\in\RR_*$ and $t(1-t)=\i$. We define the following meromorphic function as the {\em spectral function}
\begin{equation}
S_\alpha(s)=\frac{1}{\alpha}+\lim_{z\to z_0}\lbrace G^\Gamma_s(z,z_0)-\Re G^\Gamma_t(z,z_0)\rbrace.
\end{equation}
And for $\Re s>1$ we have the expression (cf. \cite{U}, Prop. 4 which holds for any discrete subgroup)
\begin{equation}
S_\alpha(s)=\frac{1}{\beta}+m_\Gamma\psi(s)+\sum_{\gamma\in\Gamma\setminus\scrI}G_s(z_0,\gamma z_0)
\end{equation}
\end{defn}

We proceed with a theorem which locates the zeros and poles of the spectral function and states their interpretation as eigenvalues or resonances of the operators $-\Delta_\varphi$ and $-\Delta$.
\begin{thm}\label{prop}
$S_\alpha(s)$ has the following zeros and poles.
\begin{itemize}
\item[(i)] $\Re s>1$: There is a zero of order one at $\tfrac{1}{2}+(\tfrac{1}{4}-\lambda_{-M}^{\alpha})^{1/2}$ provided the lowest perturbed eigenvalue satisfies $\lambda_{-M}^{\alpha}<0$.
\item[(ii)] $\tfrac{1}{2}<\Re s\leq1$: We have zeros of order one and simple poles in $\left(\tfrac{1}{2},1\right]$ corresponding to perturbed and unperturbed eigenvalues $<\tfrac{1}{4}$.
\item[(iii)] $\Re s=\tfrac{1}{2}$: There are simple poles corresponding to unperturbed eigenvalues $\geq\tfrac{1}{4}$ and a pole at $s=\tfrac{1}{2}$ which is of order 2 if $\lambda=\tfrac{1}{4}$ is in the discrete spectrum, otherwise it is simple. We also have zeros of order one corresponding to simultaneous solutions of $E(z_{0},\tfrac{1}{2}+\i t)=0$ and $\Re\left\{S_\alpha(\tfrac{1}{2}+\i t)\right\}=0$, i. e. new eigenvalues.
\item[(iv)] $\Re s<\tfrac{1}{2}$: We have zeros corresponding to perturbed resonances and poles corresponding to unperturbed resonances. There are also simple poles in $[0,\tfrac{1}{2})$ which correspond to the small eigenvalues as in $(ii)$.
\end{itemize}
\end{thm}
\begin{proof}
It can be seen from its meromorphic continuation that the regularised automorphic Green's function $(G^{\Gamma}_{s}-G^{\Gamma}_{t})(z,w)$ has simple poles corresponding to the eigenvalues of the Laplacian on the critical line and the interval $\left[0,1\right]$, as well as poles corresponding to the resonances of the Laplacian in $\Re s<\tfrac{1}{2}$. If $\lambda=\tfrac{1}{4}$ is an eigenvalue then it follows that $\G^{\Gamma}_{s}(z,w)$ has a pole of order 2 at $s=\tfrac{1}{2}$ and otherwise a simple pole (cf. \cite{Iw}, p. 106 bottom). The function $S_\alpha(s)$ inherits these poles. See for instance Thm. 3.5, p. 250, in \cite{Hj3}.

We proceed with locating the zeros of $S_\alpha(s)$. The zeros of $S_\alpha(s)$ correspond to eigenfunctions of $\Delta_\varphi$ (cf. \cite{U}, p. 10, Prop. 3). If $\Re s\geq\tfrac{1}{2}$ the zeros correspond to new eigenvalues and the corresponding eigenfunctions are in $L^{2}(X)$ and so self-adjointness of $\Delta_{\varphi}$ rules out any zeros $s$ in $\Re s\geq\tfrac{1}{2}$, $\Im s\neq0$. However, if $\Re s<\tfrac{1}{2}$ the zeros correspond to resonances and generalised non-$L^{2}$ eigenfunctions which may have $\Im s\neq0$.

Let us consider the critical line. We recall that $S_\alpha(s)$ satisfies the functional equation
\begin{equation}\label{FE}
S_\alpha(s)=S_\alpha(1-s)-\frac{E(z_{0},s)E(z_{0},1-s)}{1-2s}.
\end{equation}
Now letting $s=\tfrac{1}{2}+\i t$, $t\in\RR_*$, we compute the imaginary part of $S_\alpha(\tfrac{1}{2}+\i t)$
\begin{equation}
\begin{split}
\Im\left\{S_\alpha(\tfrac{1}{2}+\i t)\right\}
=&\frac{1}{2\i}\left\{S_\alpha(\tfrac{1}{2}+\i t)-S_\alpha(\tfrac{1}{2}-\i t)\right\}\\
=&\frac{\left|E(z_{0},\tfrac{1}{2}+\i t)\right|^{2}}{4t}
\end{split}
\end{equation}
and write
\begin{equation}
S_\alpha(\tfrac{1}{2}+\i t)=\Re\left\{S_\alpha(\tfrac{1}{2}+\i t)\right\}+\i\frac{\left|E(z_{0},\tfrac{1}{2}+\i t)\right|^{2}}{4t}
\end{equation}
so, since $\alpha$ is real, $S_\alpha(s)$ has a zero $\tfrac{1}{2}+\i t\neq\tfrac{1}{2}$ on the critical line if and only if
\begin{equation}\label{speccond}
E(z_{0},\tfrac{1}{2}+\i t)=0, \qquad \Re\left\{S_\alpha(\tfrac{1}{2}+\i t)\right\}=0.
\end{equation}
Let $\rho_\alpha$ be a solution of the above equation. We recall the spectral expansion of the automorphic Green's function for $\Re s>\tfrac{1}{2}$
\begin{equation}
\begin{split}
G^{\Gamma}_{\frac{1}{2}+\i\rho}(z,w)=&\frac{1}{2\pi}\int_{0}^{\infty}\frac{E(z,\tfrac{1}{2}+\i\rho')E(w,\tfrac{1}{2}-\i\rho')}{\rho'^{2}-\rho^{2}}d\rho'\\
&+\sum_{j=-M}^{\infty}\frac{\varphi_{j}(z)\overline{\varphi_{j}(w)}}{\rho_{j}^{2}-\rho^{2}},
\end{split}
\end{equation}
where $\lbrace\varphi_{j}\rbrace_{j=-M}^{\infty}$ is an orthonormal basis of Maass forms. Since $E(z_{0},\tfrac{1}{2}+\i\rho_\alpha)=0$, we have
\begin{equation}
\begin{split}
&\frac{d}{d\rho}\Bigg|_{\rho=\rho_\alpha}\Re\left\{S_\alpha(\tfrac{1}{2}+\i\rho)\right\}\\
=\;&\frac{\rho_\alpha}{2\pi}\int_{-\infty}^{\infty}\frac{|E(z_{0},\tfrac{1}{2}+\i\rho')|^{2}}{(\rho'^{2}-{\rho_\alpha}^2)^{2}}d\rho'
+2\rho_\alpha\sum_{j=-M}^{\infty}\frac{|\varphi_{j}(z_{0})|^{2}}{(\rho_{j}^{2}-\rho_\alpha^{2})^{2}},
\end{split}
\end{equation}
which shows that $\rho_\alpha,$ must be a zero of order one, since $\rho_\alpha\neq0$.

Now let us look at zeros on the real line. Let $\rho=-\i v$, $v>0$ to ensure $s=\tfrac{1}{2}+v>\tfrac{1}{2}$. We obtain for the derivative of $S_\alpha(\frac{1}{2}+v)$ with respect to the variable $v$
\begin{equation}
\begin{split}
\frac{d}{dv}S_\alpha(\frac{1}{2}+v)=&-\frac{v}{\pi}\int_{0}^{\infty}\frac{\left|E(z_{0},\tfrac{1}{2}+\i\rho')\right|^{2}}{(\rho'^{2}+v^{2})^{2}}d\rho'\\
&-2v\sum_{j=-M}^{\infty}\frac{\left|\varphi_{j}(z_{0})\right|^{2}}{(\rho_{j}^{2}+v^{2})^{2}}.
\end{split}
\end{equation}
Since $S_\alpha(\frac{1}{2}+v)$ is a real function on the real line, has poles in the interval $[0,\tfrac{1}{2}]$ and is monotonic in between these poles, we conclude that it must take its zeros in between these poles. If $\lim_{v\to+\infty}|S_\alpha(\tfrac{1}{2}+v)|=\infty$, then there is a zero of order one at $\rho^{\alpha}_{-M}>\rho_{-M}$. Because of monotonicity of $S_\alpha(\tfrac{1}{2}+v)$ on the half-line $(\tfrac{1}{2},\infty)$ we conclude that there are only finitely many zeros in $\Re s>\tfrac{1}{2}$. They correspond to perturbed eigenvalues less than $\tfrac{1}{4}$. As we will see later the corresponding eigenfunctions, which are automorphic Green's functions, are residues of the perturbed Eisenstein series unless $E(z_{0},\tfrac{1}{2}+v)=0$.

With regard to zeros in $(-\infty,\tfrac{1}{2})$ no such argument works, since the spectral expansion is not valid. Such zeros correspond to perturbed resonances if $E(z_{0},s)\neq0$ since the associated Green's functions fail to be in $L^{2}(X)$. If $E(z_{0},s)=0$ the zero corresponds to a zero $1-s\in(\tfrac{1}{2},\infty)$ since the functional equation for $S_\alpha(s)$ implies $S_\alpha(s)=S_\alpha(1-s)$.
\end{proof}

Colin de Verdiere observes in \cite{CdV} that for almost all pairs $(\alpha,z_{0})\in\RR_*\times X$ the new cuspidal part of the perturbed discrete spectrum is empty. We easily confirm this from the particular form of the spectral function $S_\alpha(s)$ on the critical line.
\begin{prop}\label{finite}
The equation
\begin{equation}\label{cond}
E(z_{0},\tfrac{1}{2}+\i r)=0,\qquad \Re\left\{S_\alpha(\tfrac{1}{2}+\i r)\right\}=0,\qquad r\in\RR,
\end{equation}
has no solutions for almost all $\alpha\in\RR_*$.
\end{prop}
\begin{proof}
Fix $z_{0}\in X$. Let 
\begin{equation}
S_{z_0}=\lbrace r\in\RR\mid E(z_{0},\tfrac{1}{2}+\i r)=0\rbrace
\end{equation}
which is countable since $E(z_{0},s)$ is meromorphic in $s$. Also let
\begin{equation}
g(r)=\Re\lim_{z\to z_{0}}\lbrace G^{\Gamma}_{1/2+\i r}-\tfrac{1}{2}G^{\Gamma}_{t}-\tfrac{1}{2}G^{\Gamma}_{\bar{t}}\rbrace(z,z_{0}).
\end{equation}
We define
\begin{equation}
A_{z_{0}}=\lbrace-g(r)^{-1}\mid r\in S_{z_{0}},\, g(r)\neq0\rbrace,
\end{equation}
and claim that for any $\alpha\notin A_{z_{0}}\cup\lbrace0\rbrace$ equation \eqref{cond} has no solution. To see this suppose the contrary. So there exists $r_{0}\in\RR$ such that $E(z_{0},\tfrac{1}{2}+\i r_{0})=0$ and $g(r_{0})=-\alpha^{-1}\neq0$ which leads to a contradiction since $\alpha\notin A_{z_0}$. The result follows since $A_{z_{0}}$ is countable.
\end{proof}

\section{Perturbed Eisenstein series}

The Eisenstein series $E(z,s)$ are generalised eigenfunctions of the Laplacian on such a hyperbolic surface with one cusp. In analogy with this definition we introduce an analogue of the Eisenstein series for the operator $\Delta_\varphi$. In order to admit non-$L^{2}$ eigenfunctions we enrich the domain $D_{\varphi}$ of $\Delta_{\varphi}$ by introducing the larger space
\begin{equation}
D^{*}_{0}=C^\infty_0(X\setminus\lbrace z_0\rbrace)
\end{equation}
where we have dropped the condition of square-integrability. We introduce the larger domain
\begin{equation}
D^{*}_{\varphi(\alpha)}=\lbrace g+cG^{\Gamma}_{\,t}(\cdot,z_{0})+c\,\e^{\i\varphi(\alpha)}G^{\Gamma}_{\,\bar{t}}(\cdot,z_{0})|(g,c)\in\D^{*}_{0}\times\CC\rbrace
\end{equation}
As before $\Delta_{\varphi}$ acts on $f\in D^{*}_{\varphi}$ as
\begin{equation}
\Delta_{\varphi}f=\Delta g - c\i G^{\Gamma}_{t}(z,z_{0})
+c\i\e^{\i\varphi}G^{\Gamma}_{\bar{t}}(z,z_{0}).
\end{equation}

We define the perturbed Eisenstein series on $X$ to be the solution to 
\begin{equation}\label{Eisenstein}
(\Delta_\varphi+s(1-s))E^\alpha(z,s)=0,
\end{equation}
for any $s\in\CC$ with
\begin{equation}\label{dom}
E^\alpha(\cdot,s)\in D^{*}_\varphi.
\end{equation}
In the cusp we have the asymptotics
\begin{equation}\label{asmpt}
E^\alpha(x+\i y,s)=y^{s}+\varphi_\alpha(s)y^{1-s}+O(\e^{-2\pi y}),\qquad y\to\infty,
\end{equation}
for some $\varphi_\alpha(s)\in\CC$ which will be a meromorphic function in $s$.

We now derive $E^\alpha(z,s)$ and prove a perturbative analogue of the functional equation \eqref{Greenfe} which explains the connection between $E^\alpha(z,s)$ and the automorphic Green's function.
\begin{lem}
The perturbed Eisenstein series is given by
\begin{equation}\label{eisen}
E^\alpha(z,s)=E(z,s)-\frac{E(z_{0},s)}{S_\alpha(s)}\,G^{\Gamma}_{s}(z,z_{0})
\end{equation}
and satisfies 
\begin{equation}\label{Greenfe2}
G^{\Gamma}_{1-s}(z,z_{0})=\theta_\alpha(s)G^{\Gamma}_{s}(z,z_{0})-\frac{E(z_{0},s)E^\alpha(z,1-s)}{1-2s},
\end{equation}
where
\begin{equation}
\theta_\alpha(s)=\frac{S_\alpha(1-s)}{S_\alpha(s)}.
\end{equation}
\end{lem}
\begin{proof}
Let $s\in\CC$ and $S_\alpha(s)\neq0$. First of all we show that $E^\alpha(z,s)$ as given by \eqref{eisen} is in $D^{*}_\varphi$. We define
\begin{equation}
S_\alpha(z,s)=\frac{1}{1+\e^{\i\varphi}}\lbrace G_s^\Gamma-G^{\Gamma}_{t}\rbrace(z,z_{0})
+\frac{\e^{\i\varphi}}{1+\e^{\i\varphi}}\lbrace G_s^\Gamma-G^{\Gamma}_{\bar{t}}\rbrace(z,z_{0}).
\end{equation}
In view of the identity \eqref{ext} a simple calculation confirms $$S_\alpha(s)=\lim_{z\to z_0}S_\alpha(z,s).$$
For the automorphic Green's function we have the decomposition
\begin{equation}
G^{\Gamma}_{s}(z,z_{0})=
S_\alpha(z,s)+\frac{1}{1+\e^{\i\varphi}}G^{\Gamma}_{t}(z,z_{0})+\frac{\e^{\i\varphi}}{1+\e^{\i\varphi}}G^{\Gamma}_{\bar{t}}(z,z_{0}).
\end{equation}
Substituting this back into \eqref{eisen} we obtain the decomposition of $E^\alpha(z,s)$
\begin{equation}
\begin{split}
E^\alpha(z,s)=&E(z,s)-E(z_{0},s)\frac{S_\alpha(z,s)}{S_\alpha(s)}\\
&-\frac{E(z_{0},s)}{S_\alpha(s)}\,\left\{\frac{1}{1+\e^{\i\varphi}}G^{\Gamma}_{t}(z,z_{0})+\frac{\e^{\i\varphi}}{1+\e^{\i\varphi}}G^{\Gamma}_{\bar{t}}(z,z_{0})\right\},
\end{split}
\end{equation}
which shows that $E^\alpha(\cdot,s)\in D^{*}_\varphi$ since
\begin{equation}
\begin{split}
&\lim_{z\to z_{0}}\left\{E(z,s)-E(z_{0},s)\frac{S_\alpha(z,s)}{S_\alpha(s)}\right\}\\
=\,&E(z_{0},s)\left\{1-\frac{\lim_{z\to z_{0}}S_\alpha(z,s)}{S_\alpha(s)}\right\}=0.
\end{split}
\end{equation}

We also see that
\begin{equation}
\begin{split}
&(\Delta_\varphi+s(1-s))E^\alpha(z,s)\\
=\,&(\Delta+s(1-s))\left\{E(z,s)-E(z_{0},s)\frac{S_\alpha(z,s)}{S_\alpha(s)}\right\}\\
&-\frac{1}{1+e^{\i\varphi}}\frac{E(z_{0},s)}{S_\alpha(s)}(s(1-s)-\i)G^{\Gamma}_{t}(z,z_{0})\\
&-\frac{\e^{\i\varphi}}{1+e^{\i\varphi}}\frac{E(z_{0},s)}{S_\alpha(s)}(s(1-s)+\i)G^{\Gamma}_{\bar{t}}(z,z_{0}),
\end{split}
\end{equation}
and because of
\begin{equation}
(\Delta+s(1-s))E(z,s)=0,
\end{equation}
and
\begin{equation}
\begin{split}
(\Delta+s(1-s))\,S_\alpha(z,s)=&\frac{1}{1+e^{\i\varphi}}(-s(1-s)+\i)G^{\Gamma}_{t}(z,z_{0})\\
&+\frac{\e^{\i\varphi}}{1+e^{\i\varphi}}(-s(1-s)-\i)G^{\Gamma}_{\bar{t}}(z,z_{0}),
\end{split}
\end{equation}
which follows from the iterated resolvent identity 
$$(\lambda\mp\i)\frac{1}{\Delta+s(1-s)}\frac{1}{\Delta\pm\i}=\frac{1}{\Delta\pm\i}-\frac{1}{\Delta+s(1-s)},$$ 
we have
\begin{equation}
(\Delta_\varphi+s(1-s))E^\alpha(z,s)=0.
\end{equation}

In order to prove \eqref{Greenfe2} we rewrite $E^\alpha(z,s)$ as
\begin{equation}
\begin{split}
E^\alpha(z,s)=&E(z,s)-\frac{E(z_{0},s)}{S_\alpha(s)}G^{\Gamma}_{s}(z,z_{0})\\
=&\left\{\frac{1-2s}{E(z_{0},1-s)}-\frac{E(z_{0},s)}{s_\alpha(s)}\right\}G^{\Gamma}_{s}(z,z_{0})\\
&-\frac{1-2s}{E(z_{0},1-s)}G^{\Gamma}_{1-s}(z,z_{0}),
\end{split}
\end{equation}
where \eqref{Greenfe} was substituted. Consequently one gets the equation
\begin{equation}
\begin{split}
&\frac{E^\alpha(z,s)E(z_{0},1-s)}{1-2s}\\
=\,&\left\{1-\frac{E(z_{0},1-s)E(z_{0},s)}{(1-2s)S_\alpha(s)}\right\}G^{\Gamma}_{s}(z,z_{0})-G^{\Gamma}_{1-s}(z,z_{0}).
\end{split}
\end{equation}
So let
\begin{equation}
\begin{split}
\theta_\alpha(s)&=1-\frac{E(z_{0},1-s)E(z_{0},s)}{(1-2s)S_\alpha(s)}\\
&=1-\frac{\lim_{z\to z_{0}}(G^{\Gamma}_{s}-G^{\Gamma}_{1-s})(z,z_{0})}{S_\alpha(s)}\\
&=\frac{S_\alpha(s)-\lim_{z\to z_{0}} (G^{\Gamma}_{s}-G^{\Gamma}_{1-s})(z,z_{0})}{S_\alpha(s)}\\
&=\frac{S_\alpha(1-s)}{S_\alpha(s)}.
\end{split}
\end{equation}
\end{proof}
\begin{remark}
We note from \eqref{eisen} that the eigenvalues corresponding to residual Maass forms have an interpretation as residues of $E^\alpha(z,s)$ provided $E(z_{0},s)\neq0$.
\end{remark}
It is now a straightforward corollary of the previous lemma to derive the functional equation and therefore the scattering coefficient $\varphi_\alpha(s,z_{0})$ for the perturbed Eisenstein series $\E^\alpha(z,s)$.
\begin{cor}
\begin{equation}
E^\alpha(z,s)=\varphi_\alpha(s)E^\alpha(z,1-s),
\end{equation}
where the scattering coefficient is given by
\begin{equation}\label{scatt}
\varphi_\alpha(s)=\varphi(s)\frac{S_\alpha(1-s)}{S_\alpha(s)}.
\end{equation}
\end{cor}
\begin{proof}
We use relation \eqref{Greenfe2} to derive the functional equation for $E^\alpha(z,s)$.\\
We have
\begin{equation}
\begin{split}
E^\alpha(z,s)=&-\frac{1-2s}{E(z_0,1-s)}\left[\theta_\alpha(s)G^\Gamma_s(z,z_0)-G^\Gamma_{1-s}(z,z_0)\right]\\
=&-\frac{1-2s}{\varphi(1-s)E(z_{0},s)}
\left[\theta_\alpha(s)G^\Gamma_s(z,z_0)-G^\Gamma_{1-s}(z,z_0)\right]\\
=&\,\varphi(s)\frac{1-2s}{E(z_0,s)}\theta_\alpha(s)
\left[\theta_\alpha(1-s)G^\Gamma_{1-s}(z,z_0)-G^\Gamma_s(z,z_0)\right]\\
=&\,\varphi(s)\theta_\alpha(s)E^\alpha(z,1-s)
\end{split}
\end{equation}
which is the desired result.
\end{proof}

\section{Bounding the spectral function}

We continue with the construction of a sequence of line segments $\lbrace[T_{N},T_{N}-\i\sigma]\rbrace_{N\in\NN}$ crossing the strip $\RR\times[0,-\i\sigma]$ on which the relative zeta function $S_\alpha(\tfrac{1}{2}+\i\rho)$ admits a uniform upper bound of order $e^{c(\Gamma,\alpha,z_{0})T_{N}^{2}\ln T_{N}}$ for some positive constant $c(\Gamma,\alpha,z_{0})$. We suspect that at least in the case of arithmetic groups it is possible to do much better and in fact achieve a polynomial bound. However, for a generic non-compact surface with one cusp the general bound  on the respective Eisenstein series (cf. Thm. 12.9(d) in \cite{Hj3}) gives the above bound which will suffice for our purposes. 

As in the compact case the proof uses the spectral expansion of $S_\alpha(s)$. The discrete part of the relative zeta function admits a uniform bound of polynomial growth for a suitably chosen sequence of line segments and the proof follows exactly the same lines as the proof of Proposition 9 in \cite{U}. The case of the existence of only finitely many cusp forms is trivial.

In order to obtain the analogous bound for the continuous part we require a bound on the scattering coefficient $\varphi(s)$ close to the critical line.
\begin{lem}\label{scatbound}
Let $\xi=s-\tfrac{1}{2}$ and $\sigma\geq\tfrac{1}{2}$. If $\;0\leq-\Re\xi\leq\sigma$ and $|\Im\xi|\to\infty$ then there exists a positive constant $C_{1}(\Gamma)$ s. t.
\begin{equation}
|\varphi(s)|\ll_{\Gamma,\sigma} e^{C_{1}(\Gamma)(\sigma+|\Im\xi|)}\prod_{\gamma_{j}<2|Im\xi|+2}\left|\frac{\xi-\eta_{j}-\i\gamma_{j}}{\xi+\eta_{j}+\i\gamma_{j}}\right|\left|\frac{\xi-\eta_{j}+\i\gamma_{j}}{\xi+\eta_{j}-\i\gamma_{j}}\right|.
\end{equation}
\end{lem}
\begin{proof}
The scattering coefficient $\varphi(s)$ is (see \cite{Hj3}, Prop. 12.6, p. 157) of the general form
\begin{equation}\label{prod}
\varphi(s)=\varphi(\tfrac{1}{2})e^{A\xi}\prod_{k=1}^{M}\frac{\xi-\i\rho_{-k}}{\xi+\i\rho_{-k}}\prod_{j=0}^{\infty}\frac{\xi-\eta_{j}-\i\gamma_{j}}{\xi+\eta_{j}+\i\gamma_{j}}\;\frac{\xi-\eta_{j}+\i\gamma_{j}}{\xi+\eta_{j}-\i\gamma_{j}},\;A>0,
\end{equation}
where the infinite product is convergent whenever $\xi\neq\eta_{j}\pm\i\gamma_{j}$. For all $j$ we have $\eta_{j}>0$ and $\gamma_{j}\geq0$. We also have (cf. Prop. 12.5, p. 156 \cite{Hj3})
\begin{equation}
\sum_{j=0}^{+\infty}\frac{\eta_{j}}{1+\gamma_{j}^{2}}<+\infty.
\end{equation}
The first two terms are trivially bounded. We split the infinite product into two parts. The first term is estimated as follows (cf. \cite{Hj3}, p. 159, (**) and the fourth line from the bottom)
\begin{equation}
\begin{split}
&\left|\log\prod_{\gamma_{j}\geq2|Im\xi|+2}\left|\frac{\xi-\eta_{j}-\i\gamma_{j}}{\xi+\eta_{j}+\i\gamma_{j}}\;\frac{\xi-\eta_{j}+\i\gamma_{j}}{\xi+\eta_{j}-\i\gamma_{j}}\right|\right|\\
\leq&\sum_{\gamma_{j}\geq2|Im\xi|+2}\left|\log\left(1-\frac{4\xi\eta_{j}}{(\eta_{j}+\xi)^{2}+\gamma_{j}^{2}}\right)\right|.
\end{split}
\end{equation}
We note that for $|\Im\xi|$ large we have
\begin{equation}
\begin{split}
\frac{4|\xi|\eta_{j}}{|(\eta_{j}+\xi)^{2}+\gamma_{j}^{2}|}
&\leq\frac{4|\xi|\eta_{j}}{|\Re\lbrace(\eta_{j}+\xi)^{2}+\gamma_{j}^{2}\rbrace|}\\
&\leq\frac{4|\xi|\eta_{j}}{(\eta_{j}+\Re\xi)^{2}+\gamma_{j}^{2}-(\Im\xi)^{2}}\\
&\leq\frac{4|\xi|\eta_{j}}{\gamma_{j}^{2}-(\Im\xi)^{2}}
\leq\frac{4|\xi|\eta_{j}}{\tfrac{3}{4}\gamma_{j}^{2}+1}\\
&\leq\frac{4|\xi|\eta_{j}}{3(\Im\xi)^{2}+4}\leq\frac{4(\sigma+|\Im\xi|)}{3(\Im\xi)^{2}+4}<\tfrac{1}{2}
\end{split}
\end{equation}
where we have used that $\gamma_{j}\geq2|\Im\xi|+2$ implies $\tfrac{1}{4}\gamma_{j}^{2}\geq(\Im\xi)^{2}+1$. Note that for $|z|<\tfrac{1}{2}$ we have
\begin{equation}
|\log(1-z)|\leq\sum_{k=1}^{\infty}\frac{|z|^{k}}{k}\leq2|z|,
\end{equation}
and as a consequence we have the estimate
\begin{equation}
\begin{split}
&\sum_{\gamma_{j}\geq2|Im\xi|+2}\left|\log\left(1-\frac{4\xi\eta_{j}}{(\eta_{j}+\xi)^{2}+\gamma_{j}^{2}}\right)\right|\\
&\leq8|\xi|\sum_{\gamma_{j}\geq2|Im\xi|+2}\frac{\eta_{j}}{|(\eta_{j}+\xi)^{2}+\gamma_{j}^{2}|}\\
&\leq8|\xi|\sum_{\gamma_{j}\geq2|Im\xi|+2}\frac{\eta_{j}}{1+\tfrac{3}{4}\gamma_{j}^{2}}\\
&\leq\frac{32}{3}|\xi|\sum_{j=0}^{+\infty}\frac{\eta_{j}}{1+\gamma_{j}^{2}}.
\end{split}
\end{equation}
Let $C_{1}(\Gamma)=\tfrac{32}{3}\sum_{j=0}^{+\infty}\frac{\eta_{j}}{1+\gamma_{j}^{2}}$. Then
\begin{equation}
\prod_{\gamma_{j}\geq2|Im\xi|+2}\left|\frac{\xi-\eta_{j}-\i\gamma_{j}}{\xi+\eta_{j}+\i\gamma_{j}}\;\frac{\xi-\eta_{j}+\i\gamma_{j}}{\xi+\eta_{j}-\i\gamma_{j}}\right|\leq e^{C_{1}(\Gamma)|\xi|}\leq e^{C_{1}(\Gamma)(\sigma+|\Im\xi|)}.
\end{equation}
\end{proof}

To be able to control a residual term which, as we will see, arises from the continous part of the spectral function we also require a bound on the Eisenstein series on subsets of the strip $\tfrac{1}{2}\leq\Re s\leq\tfrac{3}{2}$.
\begin{lem}\label{Ebound}
Let $c_{0}>0$. For any $T\in\RR$ there exists $T_{0}$ with $|T-T_{0}|\leq c_{0}T_{0}^{-2}$ such that
\begin{equation}
|E(z_{0},\tfrac{1}{2}+t+\i T_{0})|\ll T_{0}^{2}e^{3T_{0}}
\end{equation}
for all $t\in[0,1]$.
\end{lem}
\begin{proof}
We have the following bounds (cf. Thm. 12.9.(d), p. 164, Prop. 12.7, p. 161 \cite{Hj3}, (10.13), p. 142 \cite{Iw})
\begin{equation}\label{bound1}
|E(z_{0},\sigma+\i t)|\ll\sqrt{w(t)}\,e^{3|t|},
\end{equation}
uniformly for $\tfrac{1}{2}\leq\sigma\leq\tfrac{3}{2}$, where $w$ satisfies 
\begin{equation}
\forall t\in\RR:\,w(t)\geq1
\end{equation}
and, making use of (10.13) in \cite{Iw} and running through the same argument as on pp. 161-62 in \cite{Hj3}, $w$ furthermore satisfies
\begin{equation}\label{bound2}
\int^{R}_{-R}w(t)dt\ll R^{2}.
\end{equation}
Using bound \eqref{bound2} we infer
\begin{equation}
\int_{T-c_{0}T^{-2}}^{T+c_{0}T^{-2}}w(r)dr\leq c_{1}T^{2}
\end{equation}
for some uniform constant $c_{1}>0$. Now suppose for a contradiction that for all $r$ with $|r-T|\leq c_{0}T^{-2}$ we have $w(r)>\frac{c_{1}c_{0}^{-1}}{2}T^{4}$. This implies
\begin{equation}
\int_{T-c_{0}T^{-2}}^{T+c_{0}T^{-2}}w(r)dr>2c_{0}T^{-2}\,\frac{c_{1}c_{0}^{-1}}{2}T^{4}=c_{1}T^{2}
\end{equation}
which is a contradiction to the above bound. So we conclude that there exists $T_{0}$ with $|T-T_{0}|\leq c_{0}T^{-2}$ such that $w(T_{0})\leq\frac{c_{1}c_{0}^{-1}}{2}T^{4}$. This implies by bound \eqref{bound1}
\begin{equation}
|E(z_{0},\tfrac{1}{2}+t+\i T_{0})|\ll T_{0}^{2}e^{3T_{0}}
\end{equation}
for $t\in[0,1]$.
\end{proof}

By using Lemmas \ref{scatbound} and \ref{Ebound} we can establish a bound on the relative zeta function on intervals crossing the strip $|\Im\rho|\leq\sigma$ in between its poles and zeros. We bound the sum over the discrete spectrum analogously as in \cite{U} (cf. Proposition 9, p. 17).
\begin{prop}\label{seqbound}
There exists a sequence $\lbrace T_{N}\rbrace_{N\in\NN}$ in $\RR_{+}$, $\lim_{N\to\infty}T_{N}=+\infty$, such that for all $N$ and $t\in[0,\sigma]$ we have the uniform bound
\begin{equation}
|S_\alpha(\tfrac{1}{2}+\i\rho_{N}(t))|\ll e^{C_{2}(\Gamma)T_{N}^{2}\ln T_{N}},
\end{equation}
where $\rho_{N}(t)=T_{N}-\i t$ and $C_{2}(\Gamma)$ is some positive constant.
\end{prop}
\begin{proof}
Recall that for $\Im\rho<0$ the relative zeta function $S_\alpha(\tfrac{1}{2}+\i\rho)$ is given by
\begin{equation}
\begin{split}
S_\alpha(\tfrac{1}{2}+\i\rho)=&\alpha^{-1}+\sum_{j}|\varphi_{j}(z_{0})|^{2}\left(\frac{1}{\rho_{j}^{2}-\rho^{2}}-\Re\left[\frac{1}{\rho_{j}^{2}-\rho(\eta)^{2}}\right]\right)\\
&+\frac{1}{4\pi}\int_{-\infty}^{+\infty}E(z_{0},\tfrac{1}{2}+\i r)E(z_{0},\tfrac{1}{2}-\i r)\\
&\quad\quad\quad\quad\left(\frac{1}{r^{2}-\rho^{2}}-\Re\left[\frac{1}{r^{2}-\rho(\eta)^{2}}\right]\right)dr.
\end{split}
\end{equation}
$S_\alpha(\tfrac{1}{2}+\i\rho)$ can be continued meromorphically to the full complex plane by shifting the contour of integration and collecting a residue. We will require a continuation to the real line in order to establish the desired bound. We first have to introduce some notation.

Define $\gamma:[-1,1]\mapsto\RR_{-}$
\begin{equation}
\gamma(t)=
\begin{cases}
\begin{split}
&-1, &\text{if}\;t&\in[-\tfrac{1}{3},\tfrac{1}{3}],\\
&-1+\tfrac{3}{2}(t-\tfrac{1}{3}), &\text{if}\;t&\in[\tfrac{1}{3},1],\\
&-1+\tfrac{3}{2}(\tfrac{1}{3}+t), &\text{if}\;t&\in[-1,-\tfrac{1}{3}].
\end{split}
\end{cases}
\end{equation}
We define a contour $\Gamma_{N}(t)$, where $\rho_{N}(t)=T_{N}-\i t$, $t\in[0,\sigma]$, by the parametrisation
\begin{equation}
\Gamma_{N}(t)(u)=
\begin{cases}
\begin{split}
u,\qquad&\text{if $|u-T_{N}|>\tfrac{1}{6}cT_{N}^{-1}$}\\ 
&\text{or if $|u-T_{N}|\leq\tfrac{1}{6}cT_{N}^{-1}$ and $t\geq\tfrac{1}{12}cT_{N}^{-1}$},\\
\\
u+\i T_{N}^{-1}&\gamma(6c^{-1}T_{N}(u-T_{N})),\\
\\
&\text{if $|u-T_{N}|\leq \tfrac{1}{6}cT_{N}^{-1}$
and $0\leq t<\tfrac{1}{12}cT_{N}^{-1}$},
\end{split}
\end{cases}
\end{equation}
for $u\in\RR_{+}$ and $c>0$ is a constant which we will determine later on in the proof. To simplify notation we will also denote this decomposition by
\begin{equation}
\Gamma_{N}(t)=(\RR_{+}\backslash[T_{N}-\tfrac{1}{6}cT_{N}^{-1},T_{N}+\tfrac{1}{6}cT_{N}^{-1}])\cup\gamma_{N}(t).
\end{equation}

Let 
\begin{equation}
\delta_{N}(t)=
\begin{cases}
1,\;\text{if}\;0\leq t<\tfrac{1}{12}cT_{N}^{-1}\\
0,\;\text{if\,}\; t\geq\tfrac{1}{12}cT_{N}^{-1}.
\end{cases}
\end{equation}
Then we have
\begin{equation}\label{rep}
\begin{split}
&S_\alpha(\tfrac{1}{2}+\i\rho_{N}(t))\\
=\;&\alpha^{-1}+\sum_{j}|\varphi_{j}(z_{0})|^{2}\left(\frac{1}{\rho_{j}^{2}-\rho_{N}(t)^{2}}-\Re\left[\frac{1}{\rho_{j}^{2}-\rho(\eta)^{2}}\right]\right)\\
&-\delta_{N}(t)\frac{E(z_{0},\tfrac{1}{2}+\i\rho_{N}(t))E(z_{0},\tfrac{1}{2}-\i\rho_{N}(t))}{2\i\rho_{N}(t)}\\
&+\frac{1}{2\pi}\int_{\Gamma_{N}(t)}E(z_{0},\tfrac{1}{2}+\i r)E(z_{0},\tfrac{1}{2}-\i r)\\ &\qquad\qquad\left(\frac{1}{r^{2}-\rho_{N}(t)^{2}}-\Re\left[\frac{1}{r^{2}-\rho(\eta)^{2}}\right]\right)dr.
\end{split}
\end{equation}

Recall that the set $\lbrace\gamma_{j}\rbrace_{j=0}^{\infty}$ denotes the real parts of the resonances of $E(z,\tfrac{1}{2}+\i\rho)$ (since $i\rho=\xi=s-\tfrac{1}{2}$). Let $K=\lbrace\gamma_{j}\rbrace_{j=0}^{\infty}\cup\lbrace\rho_{i}\rbrace_{i=0}^{\infty}
=\lbrace\kappa_{l}\rbrace_{l=0}^{\infty}$. One has the upper bound (cf. (7.11), p. 101 \cite{Iw})
\begin{equation}\label{resbound}
\#\lbrace j\mid \gamma_{j}\leq R\rbrace \ll R^{2},
\end{equation}
which implies
\begin{equation}
\#\lbrace l\mid\kappa_{l}\leq R\rbrace=\#\lbrace j\mid \gamma_{j}\leq R\rbrace+\#\lbrace i\mid \rho_{i}\leq R\rbrace\ll R^{2}.
\end{equation}
Therefore we can pick an infinite sequence $\tau_{N}=\tfrac{1}{2}(\kappa_{N}+\kappa_{N+1})$ with $|\kappa_{N}-\kappa_{N+1}|\geq c_{1}\tau_{N}^{-1}$. Take $c_{0}=\tfrac{1}{6}c_{1}$ in Lemma \ref{Ebound}. So for every $N$ there exists $T_{N}$ with $|T_{N}-\tau_{N}|\leq\tfrac{1}{6}c_{1}T_{N}^{-1}$ such that
\begin{equation}\label{bound3}
|E(z_{0},\tfrac{1}{2}+t+\i T_{N})|\ll T_{N}^{2}e^{3T_{N}}.
\end{equation}
We also note $|\kappa_{N}-\kappa_{N+1}|\geq cT_{N}^{-1}$ for some $0<c<c_{1}$ and $T_{N}$ sufficiently large. This implies for all $j$
\begin{equation}
|\rho_{N}(t)-\kappa_{j}|\geq|T_{N}-\kappa_{j}|\geq\tfrac{1}{3}|\kappa_{N}-\kappa_{N+1}|\geq\tfrac{1}{3}cT_{N}^{-1}.
\end{equation}

We now turn to estimating $S_\alpha(\tfrac{1}{2}+\i\rho_{N}(t))$ where we use the representation \eqref{rep}. We split the integral into two parts, where we can drop the real part in the regularisation term for simplicity, 
\begin{equation}
\begin{split}
&\int_{\Gamma_{N}(t)}E(z_{0},\tfrac{1}{2}+\i r)E(z_{0},\tfrac{1}{2}-\i r)\left(\frac{1}{r^{2}-\rho_{N}(t)^{2}}-\frac{1}{r^{2}-\rho(\eta)^{2}}\right)dr\\
=&\left\{\int_{|r-T_{N}|>cT_{N}^{-1}/6}+\int_{\gamma_{N}(t)}\right\}E(z_{0},\tfrac{1}{2}+\i r)E(z_{0},\tfrac{1}{2}-\i r)\\
&\qquad\qquad\left(\frac{1}{r^{2}-\rho_{N}(t)^{2}}-\frac{1}{r^{2}-\rho(\eta)^{2}}\right)dr.
\end{split}
\end{equation}
We bound the tail as follows
\begin{equation}
\begin{split}
&|\rho_{N}(t)^{2}-\rho(\eta)^{2}|\int_{T_{N}+cT_{N}^{-1}/6}^{\infty}\frac{\left|E(z_{0},\tfrac{1}{2}+\i r)\right|^{2}dr}{|r^{2}-\rho_{N}(t)^{2}||r^{2}-\rho(\eta)^{2}|}\\
\ll&\,T_{N}^{2}\left\{\int_{T_{N}+\tfrac{1}{6}cT_{N}^{-1}}^{T_{N}+1}+\sum_{k=1}^{\infty}\int_{T_{N}+k}^{T_{N}+k+1}\right\}\frac{\left|E(z_{0},\tfrac{1}{2}+\i r)\right|^{2}dr}{|r^{2}-\rho_{N}(t)^{2}||r^{2}-\rho(\eta)^{2}|}\\
\ll&\,T_{N}^{2}+T_{N}^{2}\sum_{k=2}^{\infty}\frac{(T_{N}+k)^{2}}
{|(T_{N}+k)^{2}-\rho_{N}(t)^{2}||(T_{N}+k)^{2}-\rho(\eta)^{2}|}
\end{split}
\end{equation}
where we have used the mean value bound on Eisenstein series
\begin{equation}
\int_{0}^{T}|E(z_{0},\tfrac{1}{2}+\i t)|^{2}dt\ll T^{2}
\end{equation}
which follows straight away from Prop. 7.2., p. 101 in \cite{Iw}. 

Let 
\begin{equation}
b(\upsilon,\omega,\xi)=\frac{\upsilon^2}{(\upsilon^2-\omega)(\upsilon^2+\xi)}.
\end{equation}

We continue with the estimate as follows
\begin{equation}
\begin{split}
&T_{N}^{2}\left\{\sum_{k=2}^{\left\lfloor T_{N}\right\rfloor}+\sum_{k=\left\lceil T_{N}\right\rceil}^{\infty}\right\}b(T_
N+k,\rho_N(t),\tfrac{1}{4})\\
\leq \;& T_N^3\max_{2\leq k\leq\left\lfloor T_N\right\rfloor}b(T_N+k,\rho_N(t),\tfrac{1}{4})
+T_{N}^{2}\sum_{k=\left\lceil T_{N}\right\rceil}^{\infty}b(T_N+k,\rho_N(t),\tfrac{1}{4})
\end{split}
\end{equation}
It follows from $\Re\lbrace\rho_N(t)^2\rbrace=T_N^2-t^2$ that 
$$\max_{2\leq k\leq\left\lfloor T_N\right\rfloor}b(T_N+k,\rho_N(t),\tfrac{1}{4})\ll T_N^{-1}.$$
On the other hand the sum can be bounded by an integral as follows
\begin{equation}
\begin{split}
&\sum_{k=\left\lceil T_{N}\right\rceil}^\infty b(T_N+k,T_N^2-t^2,\tfrac{1}{4})\\
\ll \;&\int_{\left\lceil T_{N}\right\rceil}^\infty b(T_N+x,T_N^2-t^2,\tfrac{1}{4})dx\\
= \;&T_N^{-1}\int_{T_N^{-1}\left\lceil T_{N}\right\rceil}^\infty b(1+y,1-T_N^{-2}t^2,\tfrac{1}{4}T_N^{-2})dy\\
\ll \;& T_N^{-1}
\end{split}
\end{equation}

We obtain the following bound for the tail
\begin{equation}
|\rho_{N}(t)^{2}-\rho(\eta)^{2}|\int_{T_{N}+cT_{N}^{-1}/6}^{\infty}\frac{\left|E(z_{0},\tfrac{1}{2}+\i r)\right|^{2}dr}{|r^{2}-\rho_{N}(t)^{2}||r^{2}-\rho(\eta)^{2}|}\ll T_N^2
\end{equation}


We bound the central part by
\begin{equation}
\begin{split}
&\int_{-1}^{1}|E(z_{0},\tfrac{1}{2}+\i[\gamma_{N}(t)](r))|^{2}\\
&\qquad\frac{|\varphi(\tfrac{1}{2}-\i[\gamma_{N}(t)](r))|}
{|[\gamma_{N}(t)](r)^{2}-\rho_{N}(t)^{2}||[\gamma_{N}(t)](r)^{2}-\rho(\eta)^{2}|}\left|\frac{d[\gamma_{N}(t)](r)}{dr}\right|dt\\
\ll&\,T_{N}^{2}\,\int_{T_{N}-cT_{N}^{-1}/6}^{T_{N}+cT_{N}^{-1}/6}w(r)e^{6|r|}dr\,\sup_{r\in \,\gamma_{N}(t)}\frac{|\varphi(\tfrac{1}{2}-\i r)|}{|r^{2}-\rho_{N}(t)^{2}||r^{2}-\rho(\eta)^{2}|}\\
\ll&\,T_{N}^{4}e^{6 T_{N}}\sup_{r\in \,\gamma_{N}(t)}\frac{|\varphi(\tfrac{1}{2}-\i r)|}{|r^{2}-\rho_{N}(t)^{2}||r^{2}-\rho(\eta)^{2}|}
\end{split}
\end{equation}
which follows from \eqref{bound1} and \eqref{bound2}.

For $r\in\gamma_{N}(t)$ we have
\begin{equation}
|r-\rho_{N}(t)|\geq\tfrac{1}{12}cT_{N}^{-1}
\end{equation}
and for $T_{N}$ large
\begin{equation}
|r+\rho_{N}(t)|=T_{N}+O(T_{N}^{-1}),\quad|r^{2}-\rho(\eta)^{2}|=T_{N}^{2}+O(1).
\end{equation}
Hence
\begin{equation}
\sup_{r\in\gamma_{N}(t)}\frac{1}{|r^{2}-\rho_{N}(t)^{2}||r^{2}-\rho(\eta)^{2}|}\ll T_{N}.
\end{equation}

Let $\xi=-\i r$ and $r\in\gamma_{N}(t)$. Assume $\gamma_{j}<2|\Im\xi|+2$. Then we have
\begin{equation}
|\xi+\eta_{j}\pm\i\gamma_{j}|=|\Re\xi+\eta_{j}+\i\Im\xi\pm\i\gamma_{j}|
\geq|\Im\xi\pm\gamma_{j}|\geq\tfrac{1}{3}cT_{N}^{-1}
\end{equation}
by the choice of the sequence $\lbrace T_{N}\rbrace_{N}$. 

Let $\eta=\sup_{j}\eta_{j}$. We know that $\eta<+\infty$ (cf. \cite{Hj3}, last line in the proof of Prop. 12.5, p. 157). We thus have
\begin{equation}
\begin{split}
|\xi-\eta_{j}+\i\Im\xi\pm\i\gamma_{j}|\leq \;&|\Re\xi|+\eta_{j}+|\Im\xi|+|\gamma_{j}|\\
\leq \;&\sigma+\eta+2|\Im\xi|+2\ll T_{N}.
\end{split}
\end{equation}
This implies
\begin{equation}\label{bound4}
\begin{split}
\prod_{\gamma_{j}<2|Im\xi|+2}\left|\frac{\xi-\eta_{j}-\i\gamma_{j}}{\xi+\eta_{j}+\i\gamma_{j}}\right|\left|\frac{\xi-\eta_{j}+\i\gamma_{j}}{\xi+\eta_{j}-\i\gamma_{j}}\right|
\ll&\,T_{N}^{4\#\lbrace j\mid\gamma_{j}<2|Im\xi|+2\rbrace}\\
\leq&\,T_{N}^{4\tilde{c}(\Gamma)(2(T_{N}+cT_{N}^{-1}/6)+2)^{2}}\\
\ll&\,e^{16\tilde{c}(\Gamma)T_{N}^{2}\ln T_{N}}
\end{split}
\end{equation}
for some positive constant $\tilde{c}(\Gamma)$, where the bound in the second line follows from \eqref{resbound}. Together with Lemma \ref{scatbound} we obtain
\begin{equation}
\sup_{r\in\gamma_{N}(t)}|\varphi(\tfrac{1}{2}-\i r)|\ll e^{16cT_{N}^{2}\ln T_{N}}.
\end{equation}
For the residual term we have
\begin{equation}
\begin{split}
&\left|\frac{E(z_{0},\tfrac{1}{2}+\i\rho_{N}(t))E(z_{0},\tfrac{1}{2}-\i\rho_{N}(t))}{2\rho_{N}(t)}\right|\\
=&\;\frac{|E(z_{0},\tfrac{1}{2}+t+\i T_{N})|^{2}|\varphi(\tfrac{1}{2}-t-\i T_{N})|}{2|T_{N}-\i t|}\\
\ll&\; T_{N}^{3}e^{6T_{N}+16cT_{N}^{2}\ln T_{N}}
\end{split}
\end{equation}
where we have used \eqref{bound3} and \eqref{bound4}. Since we know that the sum in \eqref{rep} is polynomially bounded in $T_{N}$ the result follows straightaway.
\end{proof}

\section{The trace formula}

The key idea in the proof of the trace formula is to exploit the properties of the relative zeta function $S_\alpha(s)$. In particular this means the zeros and poles of $S_\alpha(s)$ which are associated with the perturbed and unperturbed discrete spectrum. We also need to make use of its link with the perturbed Eisenstein series via the scattering coefficient 
\begin{equation}
\varphi_\alpha(s)=\varphi(s)\frac{S_\alpha(1-s)}{S_\alpha(s)}.
\end{equation}
Let $B(T)=[\i\sigma,-\i\sigma]\times[-T,T]$ such that $\partial B(T)$ does not contain any zeros or poles of $S_\alpha(s)$. Similarly to the compact case the strategy in the proof of the trace formula is to first of all prove a truncated version. We then absorb the resonances into an integral along the critical line. Finally we prove the necessary bounds to justify the limit $T\to\infty$. Throughout the proof we shall assume that $\alpha\in\RR\backslash\lbrace0\rbrace$ is generic in the sense that there are only finitely many new perturbed eigenvalues (i. e. ignoring possible inherited degenerate eigenvalues of $\Delta$). Proposition \ref{finite} shows that the exceptional set is at most countable. So once we have obtained the proof of the trace formula for generic $\alpha$ we can extend the result to any $\alpha$ by continuity. Here it will be crucial that any resonances near the critical line are controlled by Proposition \ref{lim}.

Let us begin by deriving an expression involving the resonances. We can express the sum over perturbed and unperturbed resonances in the strip $0<\Im\rho<\sigma$ in terms of the following integral. We also pick up the small perturbed eigenvalues, since those are encoded in the perturbation factor of the perturbed scattering coefficient $\varphi_\alpha(s)/\varphi(s)=\theta_\alpha(s)$.
\begin{lem}\label{formula2}
Choose $\sigma$, $T$ as above. Let $h\in H_{\sigma,\delta}$. Then
\begin{equation}
\begin{split}
&-\sum_{r_{j}\in B(T)}h(r_{j})+\sum_{r^{\alpha}_{j}\in B(T)}h(r^{\alpha}_{j}) -\sum_{\rho^{\alpha}_{j}\in(0,-\i\sigma)}h(\rho^{\alpha}_{j})\\
=\,&\frac{1}{2\pi\i}\left[\int_{-\i\sigma-T}^{-\i\sigma+T}-\int_{-T}^{T}\right]h(\rho)\frac{d}{d\rho}\log\theta_\alpha(\tfrac{1}{2}+\i\rho)d\rho\\
&+\frac{1}{2\pi\i}\left[\int_{-T}^{-\i\sigma-T}+\int_{-\i\sigma+T}^{T}\right]h(\rho)\frac{d}{d\rho}\log\theta_\alpha(\tfrac{1}{2}+\i\rho)d\rho.
\end{split}
\end{equation}
\end{lem}
\begin{proof}
We prove the result by contour integration along the boundary of the box $[0,-\i\sigma]\times[-T,T]$. Since we are integrating $h$ against the logarithmic derivative of the meromorphic function $\theta_\alpha\left(\tfrac{1}{2}+\i\rho\right)$, when we shift across the contour from the line $\Im\rho=-\sigma$ to the real line we collect a residue $-kh(\rho')$ at every pole $\rho'$ of order $k$ of $\theta_\alpha\left(\tfrac{1}{2}+\i\rho\right)$ and a residue $kh(\rho')$ at every zero $\rho'$ of order $k$. The function
\begin{equation}
\theta_\alpha(\tfrac{1}{2}+\i\rho)=\frac{S_\alpha(\tfrac{1}{2}-\i\rho)}{S_\alpha(\tfrac{1}{2}+\i\rho)}
\end{equation}
has no zeros or poles corresponding to the unperturbed eigenvalues $\lbrace\rho_{j}\rbrace_{j}$ since the corresponding poles cancel because of symmetry
\begin{equation}
0=S_\alpha(\tfrac{1}{2}+\i\rho_{j})^{-1}=S_\alpha(\tfrac{1}{2}-\i\rho_{j})^{-1}
\end{equation}
where the zeros are of order one. $S_\alpha(\tfrac{1}{2}-\i\rho)$ has poles corresponding to the unperturbed resonances at $\lbrace-r_{j}\rbrace_{j}$ inside the box $[0,-\i\sigma]\times[-T,T]$ and for each of them we collect a residue $-h(-r_{j})=-h(r_{j})$. It also has zeros corresponding to perturbed resonances at $\lbrace -r^{\alpha}_{j}\rbrace_{j}$ inside the box $[0,-\i\sigma]\times[-T,T]$. For each of them we pick up a residue $h(-r^{\alpha}_{j})=h(r^{\alpha}_{j})$. The function $S_\alpha(\tfrac{1}{2}+\i\rho)^{-1}$ has simple poles in the interval $(0,-\i\sigma)$ corresponding to the small perturbed eigenvalues at $\lbrace\rho^{\alpha}_{j}\rbrace_{j=-M}^{-1}$. For each of them we pick up a residue $-h(-\rho^{\alpha}_{j})=-h(\rho^{\alpha}_{j})$. The result follows from Cauchy's residue theorem where we count multiplicities separately.
\end{proof}

We use the previous Lemma to prove a truncated trace formula. In the following let $\sigma\geq\tfrac{1}{2}$, $\delta>0$ and $h\in H_{\sigma,\delta}$.
\begin{prop}
Let $\delta_{\Gamma}=1$ if $\lambda=\tfrac{1}{4}$ is not an eigenvalue of the Laplacian and $\delta_{\Gamma}=0$ otherwise. Then we have
\begin{equation}\label{prop16}
\begin{split}
&\sum_{\rho^{\alpha}_{j}\in B(T)}h(\rho^{\alpha}_{j})-\sum_{\rho_{j}\in B(T)}h(\rho_{j})\\
=\,&\frac{1}{2\pi}\left[\int_{-\i\sigma-T}^{-\i\sigma+T}+\int_{-T}^{-T-\i\sigma}+\int_{T-\i\sigma}^{T}\right]h(\rho) \frac{S'_\alpha}{S_\alpha}(\tfrac{1}{2}+\i\rho)d\rho\\
&+\tfrac{1}{2}\delta_{\Gamma}h(0)+\frac{1}{4\pi}\int_{-T}^{T}h(\rho)\frac{\theta'_\alpha}{\theta_\alpha}(\tfrac{1}{2}+\i\rho)d\rho.
\end{split}
\end{equation}
\end{prop}
\begin{proof}
We introduce the symmetric function
\begin{equation}
\Psi(s)=S_\alpha(s)S_\alpha(1-s).
\end{equation}
and claim that
\begin{equation}
\begin{split}
&\frac{1}{\pi\i}\left[\int_{-\i\sigma-T}^{-\i\sigma+T}+\int_{-\i\sigma+T}^{\i\sigma+T}\right]h(\rho)\frac{d}{d\rho}\log\Psi(\tfrac{1}{2}+\i\rho)d\rho\\
=\,&-2\delta_{\Gamma}h(0)+4\sum_{\rho^{\alpha}_{j}\in (-T,T)}h(\rho^{\alpha}_{j})-4\sum_{\rho_{j}\in B(T)}h(\rho_{j})-2\sum_{r_{j}\in B(T)}h(r_{j})\\
&+2\sum_{\rho^{\alpha}_{j}\in (0,-\i\sigma)}h(\rho^{\alpha}_{j})+2\sum_{r^{\alpha}_{j}\in B(T)}h(r^{\alpha}_{j}).
\end{split}
\end{equation}
The above identity is easily proven by contour integration along the boundary of the box $[\i\sigma,-\i\sigma]\times[-T,T]$. We simply have to count all zeros and poles of $\Psi(s)$ that lie inside the box. From Theorem \ref{prop} we know that $S_\alpha(\tfrac{1}{2}+\i\rho)$ has simple poles at $\lbrace\rho_{j}\rbrace_{\rho_{j}\in B(T)}$ and $\lbrace-\rho_{j}\rbrace_{\rho_{j}\in B(T)}$ corresponding to eigenvalues $\lbrace\lambda_{j}\rbrace_{\rho_{j}\in B(T)}$, $\lambda_{j}=\tfrac{1}{4}+\rho_{j}^{2}$. They have the same residues, ie. $-h(\rho_{j})=-h(-\rho_{j})$ and $-h(\rho_{-j})=-h(-\rho_{-j})$, because of evenness of $h$. So we collect residues $-2\sum_{\rho_{j}\in B(T)}h(\rho_{j})$ altogether. We also collect residues $-\sum_{r_{j}\in B(T)}h(r_{j})$ at the unperturbed resonances $\lbrace r_{j}\rbrace_{r_{j}\in B(T)}$ that lie inside the box.

Also from Theorem \ref{prop} we know that $S_\alpha(s)$ has zeros of order one at $\lbrace\rho^{\alpha}_{j}\rbrace_{\rho^{\alpha}_{j}\in (-T,T)}$ and $\lbrace-\rho^{\alpha}_{j}\rbrace_{\rho^{\alpha}_{j}\in (-T,T)}$ corresponding to the perturbed eigenvalues $\lbrace\lambda^{\alpha}_{j}\rbrace_{\rho^{\alpha}_{j}\in (-T,T)}$, where $\lambda^{\alpha}_{j}=\tfrac{1}{4}+{\rho^{\alpha}_{j}}^{2}\geq\tfrac{1}{4}$, since the condition for the existence of a zero on the critical line implies $E(z_{0},\tfrac{1}{2}+\i\rho^{\alpha}_{j})=0$ and therefore by \eqref{FE}
\begin{equation}
0=S_\alpha(\tfrac{1}{2}+\i\rho^{\alpha}_{j})=S_\alpha(\tfrac{1}{2}-\i\rho^{\alpha}_{j}),
\end{equation}
so the residues are the same and altogether we collect $2\sum_{\rho^{\alpha}_{j}\in(-T,T)}h(\rho^{\alpha}_{j})$. 

It is crucial to note that in the case of the small perturbed eigenvalues it is not a requirement that $E(z_{0},s)$ vanishes at the corresponding point. Therefore we may have only one zero of order one corresponding to a small eigenvalue $\tfrac{1}{4}+{\rho^{\alpha}_{-j}}^{2}$ at $\rho^{\alpha}_{-j}\in(0,-\i\sigma)$. Altogether we collect residues $\sum_{\rho^{\alpha}_{j}\in(0,-\i\sigma)}h(\rho^{\alpha}_{j})$. For the perturbed resonances $\lbrace r^{\alpha}_{j}\rbrace_{r^{\alpha}_{j}\in B(T)}$, which correspond to zeros of $S_\alpha(\tfrac{1}{2}+\i\rho)$ in the halfplane $\Im\rho>0$, we collect residues $\sum_{r^{\alpha}_{j}\in B(T)}h(r^{\alpha}_{j})$. Finally, if $\lambda=\tfrac{1}{4}$ is an eigenvalue we collect a residue $-2h(0)$ at $s=\tfrac{1}{2}$ since the pole is of order 2. If $\lambda=\tfrac{1}{4}$ is not an eigenvalue, then we collect a residue $-h(0)$ and we shall write this separately in the final summation over residues, since it has no interpretation with respect to an eigenvalue. 

By evenness of $h$ and $\Psi$ we can write
\begin{equation}
\begin{split}
&\frac{1}{2\pi\i}\int_{\partial B}h(\rho)\frac{d}{d\rho}\log\Psi(\tfrac{1}{2}+\i\rho)d\rho\\
=\,&\frac{1}{\pi\i}\left[\int_{-\i\sigma-T}^{-\i\sigma+T}+\int_{-\i\sigma+T}^{\i\sigma+T}\right]h(\rho)\frac{d}{d\rho}\log\Psi(\tfrac{1}{2}+\i\rho)d\rho
\end{split}
\end{equation}
By Cauchy's residue theorem the RHS equals the sum over the residues collected. By evenness of $h$ and $\Psi(\tfrac{1}{2}+\i\rho)=S_\alpha(\tfrac{1}{2}+\i\rho) S_\alpha(\tfrac{1}{2}-\i\rho)$ the sum over the residues equals exactly twice the sums over residues listed above. This proves the above claim. We now apply the previous Lemma in order to express the resonances which appear above in terms of an integral and a finite sum over the small perturbed eigenvalues. We have
\begin{equation}
\begin{split}
&-2\sum_{r_{j}\in B(T)}h(r_{j})+2\sum_{r^{\alpha}_{j}\in B(T)}h(r^{\alpha}_{j})\\
=\;&\frac{1}{\pi\i}\left[\int_{-\i\sigma-T}^{-\i\sigma+T}-\int_{-T}^{T}+\int_{-T}^{-\i\sigma-T}+\int_{-\i\sigma+T}^{T}\right]h(\rho)\frac{d}{d\rho}\log\theta_\alpha(\tfrac{1}{2}+\i\rho)d\rho\\
&+2\sum_{\rho^{\alpha}_{j}\in (0,-\i\sigma)}h(\rho^{\alpha}_{j}).
\end{split}
\end{equation}
After substituting and dividing through by $4$ we obtain
\begin{equation}\label{expr01}
\begin{split}
&\frac{1}{4\pi\i}\left[\int_{-\i\sigma-T}^{-\i\sigma+T}+\int_{-\i\sigma+T}^{\i\sigma+T}\right]h(\rho)\frac{d}{d\rho}
\log\Psi(\tfrac{1}{2}+\i\rho)d\rho\\
=&-\tfrac{1}{2}\delta_{\Gamma}h(0)+\sum_{\rho^{\alpha}_{j}\in B(T)}h(\rho^{\alpha}_{j})
-\sum_{\rho_{j}\in B(T)}h(\rho_{j})\\
&+\frac{1}{4\pi\i}\left[\int_{-\i\sigma-T}^{-\i\sigma+T}-\int_{-T}^{T}\right]h(\rho)
\frac{d}{d\rho}\log\theta_\alpha(\tfrac{1}{2}+\i\rho)d\rho\\
&+\frac{1}{4\pi\i}\left[\int_{-T}^{-\i\sigma-T}+\int_{-\i\sigma+T}^{T}\right]h(\rho)
\frac{d}{d\rho}\log\theta_\alpha(\tfrac{1}{2}+\i\rho)d\rho
\end{split}
\end{equation}
We can rewrite
\begin{equation}
\int_{T}^{\i\sigma+T}h(\rho)\frac{d}{d\rho}\log\Psi(\tfrac{1}{2}+\i\rho)d\rho
=\int_{-T}^{-\i\sigma-T}h(\rho)\frac{d}{d\rho}\log\Psi(\tfrac{1}{2}+\i\rho)d\rho
\end{equation}
since the integrand is odd. Therefore, noting the relation 
\begin{equation}
\Psi(s)=S_\alpha(s)^{2}\theta_\alpha(s),
\end{equation}
we have
\begin{equation}
\begin{split}
&\frac{1}{4\pi\i}\int_{-\i\sigma+T}^{\i\sigma+T}h(\rho)\frac{d}{d\rho}
\log\Psi(\tfrac{1}{2}+\i\rho)d\rho\\
&-\frac{1}{4\pi\i}\left[\int_{-T}^{-\i\sigma-T}+\int_{-\i\sigma+T}^{T}\right]h(\rho)
\frac{d}{d\rho}\log\theta_\alpha(\tfrac{1}{2}+\i\rho)d\rho\\
=&\;\frac{1}{2\pi\i}\left[\int_{-T}^{-T-\i\sigma}+\int_{T-\i\sigma}^{T}\right]h(\rho)\frac{d}{d\rho}\log S_\alpha(\tfrac{1}{2}+\i\rho)d\rho.
\end{split}
\end{equation}
Similarly we have
\begin{equation}
\begin{split}
&\frac{1}{4\pi\i}\int_{-\i\sigma-T}^{-\i\sigma+T}h(\rho)\frac{d}{d\rho}
\log\Psi(\tfrac{1}{2}+\i\rho)d\rho\\
&-\frac{1}{4\pi\i}\int_{-\i\sigma-T}^{-\i\sigma+T}h(\rho)
\frac{d}{d\rho}\log\theta_\alpha(\tfrac{1}{2}+\i\rho)d\rho\\
=\;&\frac{1}{2\pi\i}\int_{-\i\sigma-T}^{-\i\sigma+T}h(\rho)\frac{d}{d\rho}\log S_\alpha(\tfrac{1}{2}+\i\rho)d\rho.
\end{split}
\end{equation}
So finally we can rewrite equation \eqref{expr01} as
\begin{equation}
\begin{split}
&\frac{1}{2\pi\i}\left[\int_{-\i\sigma-T}^{-\i\sigma+T}+\int_{-T}^{-T-\i\sigma}+\int_{T-\i\sigma}^{T}\right]h(\rho)\frac{d}{d\rho}\log S_\alpha(\tfrac{1}{2}+\i\rho)d\rho\\
=\;&-\tfrac{1}{2}\delta_{\Gamma}h(0)+\sum_{\rho^{\alpha}_{j}\in B(T)}h\left(\rho^{\alpha}_{j}\right)-\sum_{\rho_{j}\in B(T)}h\left(\rho_{j}\right)\\
&-\frac{1}{4\pi\i}\int_{-T}^{T}h(\rho)\frac{d}{d\rho}\log\theta_\alpha(\tfrac{1}{2}+\i\rho)d\rho,
\end{split}
\end{equation}
which is the desired result.
\end{proof}

For the proof of the trace formula we need to be able to control the term containing the scattering coefficients. We state this result in the following proposition.
\begin{prop}\label{lim}
For any $h\in H_{\sigma,\,\delta}$
\begin{equation}
\frac{1}{4\pi}\int_{-\infty}^{+\infty}h(\rho)\left\{\frac{\varphi_\alpha'}{\varphi_\alpha}(\tfrac{1}{2}+\i\rho)-\frac{\varphi'}{\varphi}(\tfrac{1}{2}+\i\rho)\right\}d\rho
\end{equation}
converges absolutely.
\end{prop}
\begin{proof}
We recall that the perturbed scattering coefficient is given by
\begin{equation}
\varphi_\alpha(s)=\varphi(s)\frac{S_\alpha(1-s)}{S_\alpha(s)}.
\end{equation}
Consequently,
\begin{equation}
\frac{\varphi_\alpha'}{\varphi_\alpha}(s)-\frac{\varphi'}{\varphi}(s)=\frac{d}{ds}\log\frac{S_\alpha(1-s)}{S_\alpha(s)},
\end{equation}
and therefore
\begin{equation}
\begin{split}
&\frac{1}{4\pi}\int_{-T}^{T}h(\rho)\left\{\frac{\varphi_\alpha'}{\varphi_\alpha}(\tfrac{1}{2}+\i\rho)-\frac{\varphi'}{\varphi}(\tfrac{1}{2}+\i\rho)\right\}d\rho\\
=\,&\frac{1}{4\pi\i}\int_{-T}^{T}h(\rho)\frac{d}{d\rho}\log\frac{S_\alpha(\tfrac{1}{2}-\i\rho)}{S_\alpha(\tfrac{1}{2}+\i\rho)}d\rho\\
=\,&\frac{1}{2\pi\i}h(T)\log\frac{S_\alpha(\tfrac{1}{2}-\i T)}{S_\alpha(\tfrac{1}{2}+\i T)}
-\frac{1}{4\pi\i}\int_{-T}^{T}h'(\rho)\log\frac{S_\alpha(\tfrac{1}{2}-\i\rho)}{S_\alpha(\tfrac{1}{2}+\i\rho)}d\rho.
\end{split}
\end{equation}
So in order to prove the existence of the limit as $T\to\infty$ we have to estimate
\begin{equation}
\log\frac{S_\alpha(\tfrac{1}{2}-\i\rho)}{S_\alpha(\frac{1}{2}+\i\rho)}
\end{equation}
for $\rho\in\RR$. We know that
\begin{equation}
S_\alpha(\tfrac{1}{2}+\i\rho)=r(\rho)\e^{i\theta(\rho)},
\end{equation}
where $\theta(\rho)=\arg S_\alpha(\tfrac{1}{2}+\i\rho)$ and $r(\rho)=\left|S_\alpha(\tfrac{1}{2}+\i\rho)\right|$. We recall the functional equation
\begin{equation}
S_\alpha(s)=S_\alpha(1-s)-\frac{E(z_{0},s)E(z_{0},1-s)}{1-2s},
\end{equation}
and see as a consequence
\begin{equation}
r(\rho)\sin\theta(\rho)=\Im S_\alpha(\tfrac{1}{2}+\i\rho)=\frac{1}{4\rho}\left|E(z_{0},\tfrac{1}{2}+\i\rho)\right|^{2}.
\end{equation}
This together with meromorphicity of $S_\alpha(s)$ implies $\left|\theta(\rho)\right|\leq\pi$. In fact this bound holds in the strip $-\sigma\leq\Im\rho\leq0$. To see this observe that for $\Re\rho>0$ we have
\begin{equation}\label{Im}
\begin{split}
&\Im S_\alpha(\tfrac{1}{2}+\i\rho)\\
=\;&-\frac{\Re\rho\;\Im\rho}{2\pi}\int_{0}^{+\infty}\frac{|E(z_{0},\tfrac{1}{2}+\i r)|^{2}dr}{(\Re\rho)^{2}-r^{2}-(\Im\rho)^{2})^{2}+4(\Re\rho)^{2}(\Im\rho)^{2}}\\
&-\Re\rho\;\Im\rho\sum_{j=-M}^{\infty}\frac{|\varphi_{j}(z_{0})|^{2}} {((\Re\rho)^{2}-\rho_{j}^{2}-(\Im\rho)^{2})^{2}+4(\Re\rho)^{2}(\Im\rho)^{2}}\\
\geq\;&0
\end{split}
\end{equation}
From the functional equation we see
\begin{equation}
\lim_{\Im\rho\to0}\Im S_\alpha(\tfrac{1}{2}+\i\rho)=\frac{|E(z_{0},\tfrac{1}{2}+\i\Re\rho)|^{2}}{2\Re\rho}\geq0.
\end{equation}
From this observation and since $S_\alpha(\tfrac{1}{2}+\i\rho)$ is meromorphic, it follows that $S_\alpha(\tfrac{1}{2}+\i\cdot)$ maps any smooth curve in the strip $-\sigma\leq\Im\rho\leq0$ which contains no poles of $S_\alpha(\tfrac{1}{2}+\i\rho)$ to a smooth curve in the upper halfplane. In particular this implies that the winding number of any such curve is 0. This implies $0\leq\arg S_\alpha(\tfrac{1}{2}+\i\rho)\leq\pi$ for $\Re\rho>0$ and $-\sigma\leq\Im\rho\leq0$ away from poles. In the same way one sees that $-\pi\leq\arg S_\alpha(\tfrac{1}{2}+\i\rho)\leq0$ for $\Re\rho<0$ and $-\sigma\leq\Im\rho\leq0$ away from poles.

So we have
\begin{equation}
\left|\log\frac{S_\alpha(\tfrac{1}{2}-\i\rho)}{S_\alpha(\tfrac{1}{2}+\i\rho)}\right|=\left|\log\frac{r(-\rho)\e^{-\i\theta(\rho)}}{r(\rho)\e^{\i\theta(\rho)}}\right|=2\left|\theta(\rho)\right|\leq 2\pi,
\end{equation}
where we have used
\begin{equation}
r(\rho)=|S_\alpha(\tfrac{1}{2}+\i\rho)|=|\overline{S_\alpha(\tfrac{1}{2}+\i\rho)}|
=|S_\alpha(\tfrac{1}{2}-\i\rho)|=r(-\rho).
\end{equation}
and the fact that $S_\alpha(\tfrac{1}{2}+\i\rho)S_\alpha(\tfrac{1}{2}-\i\rho)$ is analytic on the real line.

Now we can estimate the tail as follows
\begin{equation}
\begin{split}
&\frac{1}{4\pi}\left|\int_{T}^{\infty}h(\rho)\left\{\frac{\varphi_\alpha'}{\varphi_\alpha}(\tfrac{1}{2}+\i\rho)-\frac{\varphi'}{\varphi}(\tfrac{1}{2}+\i\rho)\right\}d\rho\right|\\
=\,&\frac{1}{4\pi}\left|\int_{T}^{\infty}h(\rho)\frac{d}{d\rho}\log\frac{S_\alpha(\tfrac{1}{2}-\i\rho)}{S_\alpha(\tfrac{1}{2}+\i\rho)}d\rho\right|\\
\leq\,&\frac{1}{4\pi}\left|h(T)\right|\left|\log\frac{S_\alpha(\tfrac{1}{2}-\i T)}{S_\alpha(\tfrac{1}{2}+\i T)}\right|\\
&+\frac{1}{4\pi}\left|\int_{T}^{\infty}h'(\rho)\log\frac{S_\alpha(\tfrac{1}{2}-\i\rho)}{S_\alpha(\tfrac{1}{2}+\i\rho)}d\rho\right|\\
\leq\,&\frac{1}{4\pi}\left|h(T)\right|\left|\log\frac{S_\alpha(\tfrac{1}{2}-\i T)}{S_\alpha(\tfrac{1}{2}+\i T)}\right|\\
&+\frac{1}{4\pi}\int_{T}^{\infty}|h'(\rho)|\left|\log\frac{S_\alpha(\tfrac{1}{2}-\i\rho)}{S_\alpha(\tfrac{1}{2}+\i\rho)}\right|d\rho.
\end{split}
\end{equation}
We only have to estimate one tail because the relation
\begin{equation}
\frac{\varphi'(s)}{\varphi(s)}=\frac{\varphi'(1-s)}{\varphi(1-s)}
\end{equation}
implies that the integrand is even on the critical line. As an immediate consequence of the decay of $h$ and the above bound on the logarithm we conclude
\begin{equation}
\lim_{T\to\infty}\frac{1}{4\pi}\left|h(T)\right|\left|\log\frac{S_\alpha(\tfrac{1}{2}-\i T)}{S_\alpha(\tfrac{1}{2}+\i T)}\right|=0.
\end{equation}
Since $h\in H_{\sigma,\delta}$ we have the estimate
\begin{equation}
|h'(\rho)|\ll(1+|\Re\rho|)^{-2-\delta}
\end{equation}
uniformly in the strip $|\Im\rho|\leq\sigma-\epsilon$ for any $0<\epsilon<\sigma$. So we obtain
\begin{equation}
\lim_{T\to\infty}\frac{1}{4\pi}\int_{T}^{\infty}|h'(\rho)|\left|\log\frac{S_\alpha(\tfrac{1}{2}-\i\rho)}{S_\alpha(\tfrac{1}{2}+\i\rho)}\right|d\rho=0.
\end{equation}
We conclude
\begin{equation}
\begin{split}
&\lim_{T\to\infty}\frac{1}{4\pi}\int_{-T}^{T}h(\rho)\left\{\frac{\varphi_\alpha'}{\varphi_\alpha}(\tfrac{1}{2}+\i\rho)-\frac{\varphi'}{\varphi}(\tfrac{1}{2}+\i\rho)\right\}d\rho\\
=&\frac{1}{4\pi}\int_{-\infty}^{\infty}h(\rho)\left\{\frac{\varphi_\alpha'}{\varphi_\alpha}(\tfrac{1}{2}+\i\rho)-\frac{\varphi'}{\varphi}(\tfrac{1}{2}+\i\rho)\right\}d\rho.
\end{split}
\end{equation}
\end{proof}

With the previous work we are able to justify the existence of the limit $T\to\infty$ of the truncated trace formula \eqref{prop16}.
\begin{cor}
Let $\lbrace T_{n}\rbrace$ be any sequence in $\RR_{+}$ in between eigenvalues $\lbrace\rho_{j}\rbrace_{j=0}^{+\infty}$, $\lbrace\rho_{j}^{\alpha}\rbrace_{j=0}^{+\infty}$ and real parts of resonances $\lbrace\Re r_{j}\rbrace_{j=0}^{+\infty}$ and accumulating at $+\infty$. Then the limit
\begin{equation}\label{bdterms}
\frac{1}{2\pi\i}\lim_{n\to\infty}\left\{\int_{-\i\sigma+T_{n}}^{T_{n}}+\int_{-T_{n}}^{-\i\sigma-T_{n}}\right\}h(\rho)\frac{d}{d\rho}\log S_\alpha(\tfrac{1}{2}+\i\rho)d\rho
\end{equation}
exists.
\end{cor}
\begin{proof}
We have the identity
\begin{equation}
\begin{split}
&\frac{1}{2\pi\i}\left\{\int_{-\i\sigma+T_{n}}^{T_{n}}+\int_{-T_{n}}^{-\i\sigma-T_{n}}\right\}h(\rho)\frac{d}{d\rho}\log S_\alpha(\tfrac{1}{2}+\i\rho)d\rho\\
=\;&\sum_{\rho^{\alpha}_{j}\in B(T_{n})}h(\rho^{\alpha}_{j})-\sum_{\rho_{j}\in B(T_{n})}h(\rho_{j})\\
&-\frac{1}{2\pi\i}\int_{-\i\sigma-T_{n}}^{-\i\sigma+T_{n}}h(\rho)\frac{d}{d\rho}\log S_\alpha(\tfrac{1}{2}+\i\rho)d\rho\\
&-\tfrac{1}{2}\delta_{\Gamma}h(0)-\frac{1}{4\pi}\int_{-T_{n}}^{T_{n}}h(\rho)\frac{\theta_\alpha'}{\theta_\alpha}(\tfrac{1}{2}+\i\rho)d\rho.
\end{split}
\end{equation}
The standard upper bound on the number of eigenvalues and the decay of $h$ imply that the sums over the eigenvalues converge absolutely. By Proposition \ref{lim} the limit
\begin{equation}
\lim_{T\to\infty}\frac{1}{4\pi}\int_{-T}^{T}h(\rho)\frac{\theta_\alpha'}{\theta_\alpha}(\tfrac{1}{2}+\i\rho)d\rho
=\frac{1}{4\pi}\int_{-\infty}^{\infty}h(\rho)\frac{\theta_\alpha'}{\theta_\alpha}(\tfrac{1}{2}+\i\rho)d\rho
\end{equation}
exists. Finally, for $\Im\rho=-\sigma$, we have as a consequence of Lemma 9
\begin{equation}
S_\alpha(\tfrac{1}{2}+\i\rho)=m\psi(\tfrac{1}{2}+\i\rho)+O(1)
\end{equation}
which implies, for $\Im\rho=-\sigma$ and $|\rho|$ large, using boundedness of $\arg S_\alpha(\tfrac{1}{2}+\i\rho)$ away from poles,
\begin{equation}\label{Sbound}
\log S_\alpha(\tfrac{1}{2}+\i\rho)=O(\log\log|\rho|).
\end{equation}
From an integration by parts we have
\begin{equation}
\begin{split}
&\int_{-\i\sigma-T}^{-\i\sigma+T}h(\rho)\frac{d}{d\rho}\log S_\alpha(\tfrac{1}{2}+\i\rho)d\rho\\
=\;&h(-\i\sigma+T)\log S_\alpha(\tfrac{1}{2}+\sigma+\i T)\\
&-h(-\i\sigma-T)\log S_\alpha(\tfrac{1}{2}+\sigma-\i T)\\
&-\int_{-\i\sigma-T}^{-\i\sigma+T}h'(\rho)\log S_\alpha(\tfrac{1}{2}+\i\rho)d\rho
\end{split}
\end{equation}
which implies
\begin{equation}\label{limbound}
\begin{split}
&\lim_{T\to\infty}\int_{-\i\sigma-T}^{-\i\sigma+T}h(\rho)\frac{d}{d\rho}\log S_\alpha(\tfrac{1}{2}+\i\rho)d\rho\\
=\,&\lim_{T\to\infty}\int_{-\i\sigma-T}^{-\i\sigma+T}h'(\rho)\log S_\alpha(\tfrac{1}{2}+\i\rho)d\rho
\end{split}
\end{equation}
since
\begin{equation}
\begin{split}
&\lim_{T\to+\infty}h(-\i\sigma+T)\log S_\alpha(\tfrac{1}{2}+\sigma+\i T)\\
=\,&\lim_{T\to+\infty}h(-\i\sigma-T)\log S_\alpha(\tfrac{1}{2}+\sigma-\i T)=0.
\end{split}
\end{equation}
Now the integral on the RHS of \eqref{limbound} clearly converges because of \eqref{Sbound} and $h\in H_{\sigma,\delta}$. Consequently the limit of the boundary term exists as we stretch the box to infinity.
\end{proof}

We can now apply Proposition \ref{seqbound} making use of the particular test function, which is constructed in Lemma 11 in \cite{U}, to derive a bound on a sequence of integrals which will eventually lead to the required bound on the corresponding sequence of boundary terms \eqref{bdterms} which we require in order to show that its limit vanishes.
\begin{prop}\label{prop24}
There exists $\lbrace T_{N(j)}\rbrace_{j}\subset\RR_{+}$, $\lim_{j\to\infty}T_{N(j)}=\infty$, such that for any $\epsilon>0$ we have
\begin{equation}
\int_{T_{N(j)}-\i\sigma}^{T_{N(j)}}\big|\log|S_\alpha(\tfrac{1}{2}+\i\rho)|\big||d\rho|\ll_{\epsilon}T_{N(j)}^{2+\epsilon}.
\end{equation}
\end{prop}
\begin{proof}
Consider the sequence $\lbrace T_{N}\rbrace_{N}$ of Proposition \ref{seqbound}. We choose the subsequence $\lbrace T_{N(j)}\rbrace_{j}$ as in Lemma 11 in \cite{U}. Fix $\epsilon>0$. We pick the test function $h_{\epsilon}\in H_{\sigma, \epsilon}$ given in this lemma (the explicit construction is given in the appendix of \cite{U})
and we recall its properties
\begin{equation}
\overline{h_\epsilon(\rho)}=h_\epsilon(\bar{\rho})
\end{equation}
and uniformly in $\rho\in[T_{N(j)},T_{N(j)}-\i\sigma]$
\begin{equation}
|\Re h'_\epsilon(\rho)|\gg T_{N(j)}^{-2-\epsilon}.
\end{equation}

We have by integration by parts
\begin{equation}\label{byparts}
\begin{split}
&\left\{\int_{T_{N(j)}-\i\sigma}^{T_{N(j)}}+\int_{-T_{N(j)}}^{-T_{N(j)}-\i\sigma}\right\}h_{\epsilon}(\rho)\frac{d}{d\rho}\log  S_\alpha(\tfrac{1}{2}+\i\rho)d\rho\\
=\;&h_{\epsilon}(T_{N(j)})\log S_\alpha(\tfrac{1}{2}+\i T_{N(j)})\\
&-h_{\epsilon}(T_{N(j)}-\i\sigma)\log S_\alpha
(\tfrac{1}{2}+\sigma+\i T_{N(j)})\\
&+h_{\epsilon}(-T_{N(j)}-\i\sigma)\log S_\alpha
(\tfrac{1}{2}+\sigma-\i T_{N(j)})\\
&-h_{\epsilon}(-T_{N(j)})\log S_\alpha(\tfrac{1}{2}-\i T_{N(j)})\\
&-\left\{\int_{T_{N(j)}-\i\sigma}^{T_{N(j)}}+\int_{-T_{N(j)}}^{-T_{N(j)}-\i\sigma}\right\}h_{\epsilon}'(\rho)\log S_\alpha(\tfrac{1}{2}+\i\rho)d\rho.
\end{split}
\end{equation}
We know that the LHS converges. Recall
\begin{equation}\label{log}
\begin{split}
|S_\alpha(\tfrac{1}{2}+\sigma\pm\i T_{N(j)})|&\sim\psi(\tfrac{1}{2}+\sigma\pm\i T_{N(j)})\\
&\sim\log T_{N(j)},\qquad j\to\infty.
\end{split}
\end{equation}
We have
\begin{equation}
|h_{\epsilon}(\rho)|\ll(1+|\Re\rho|)^{-2-\epsilon}
\end{equation}
uniformly for $|\Im\rho|\leq\sigma$. It follows that the second and third term converge to zero as $j\to\infty$. For the first term and very similarly for the fourth term we have
\begin{equation}
\begin{split}
&h_{\epsilon}(T_{N(j)})\log S_\alpha(\tfrac{1}{2}+\i T_{N(j)})\\
=\;&h_{\epsilon}(T_{N(j)})\log|S_\alpha(\tfrac{1}{2}+\i T_{N(j)})|\\
&+\i h_{\epsilon}(T_{N(j)})\arg S_\alpha(\tfrac{1}{2}+\i T_{N(j)}).
\end{split}
\end{equation}
We can rewrite the fifth term as
\begin{equation}
\begin{split}
&\left\{\int_{T_{N(j)}-\i\sigma}^{T_{N(j)}}+\int_{-T_{N(j)}}^{-T_{N(j)}-\i\sigma}\right\}h_{\epsilon}'(\rho)\log |S_\alpha(\tfrac{1}{2}+\i\rho)|d\rho\\
&+\i\left\{\int_{T_{N(j)}-\i\sigma}^{T_{N(j)}}+\int_{-T_{N(j)}}^{-T_{N(j)}-\i\sigma}\right\}h_{\epsilon}'(\rho)\arg S_\alpha(\tfrac{1}{2}+\i\rho)d\rho
\end{split}
\end{equation}

Recall that $\arg S_\alpha(\tfrac{1}{2}+\i\rho)$ is bounded away from poles in $-\sigma\leq\Im\rho\leq0$. This implies $|\arg S_\alpha(\tfrac{1}{2}+\i\rho)|\ll1$ for $\rho\in[T_{N}(j),T_{N}(j)-\i\sigma]$, as we recall that the sequence $\lbrace T_{N(j)}\rbrace_{j}$ is chosen such that the intervals do not contain any poles of $S_\alpha(\tfrac{1}{2}+\i\rho)$. It hence follows
\begin{equation}
|h_{\epsilon}(\pm T_{N(j)})||\arg(S_\alpha(\tfrac{1}{2}\pm\i T_{N(j)}))|\ll T_{N(j)}^{-2-\epsilon}
\end{equation}
and
\begin{equation}
\left\{\int_{T_{N(j)}-\i\sigma}^{T_{N(j)}}+\int_{-T_{N(j)}}^{-T_{N(j)}-\i\sigma}\right\}h_{\epsilon}'(\rho)\arg S_\alpha(\tfrac{1}{2}+\i\rho)d\rho\ll T_{N}^{-2-\epsilon}.
\end{equation}

Combining all this we conclude that the limit of
\begin{equation}
\begin{split}
&\lbrace h_{\epsilon}(T_{N(j)})-h_{\epsilon}(-T_{N(j)})\rbrace\log|S_\alpha(\tfrac{1}{2}+\i T_{N(j)})|\\
&-\left\{\int_{T_{N(j)}-\i\sigma}^{T_{N(j)}}+\int_{-T_{N(j)}}^{-T_{N(j)}-\i\sigma}\right\}h_{\epsilon}'(\rho)\log |S_\alpha(\tfrac{1}{2}+\i\rho)|d\rho
\end{split}
\end{equation}
which we can rewrite as
\begin{equation}\label{rearr}
\begin{split}
=\;&\int_{T_{N(j)}-\i\sigma}^{T_{N(j)}}\lbrace h'_{\epsilon}(\rho)-h'_{\epsilon}(-\bar{\rho})\rbrace\log |S_\alpha(\tfrac{1}{2}+\i\rho)|d\rho\\
=\;&-2\i\int_{0}^{\sigma}\Re h_{\epsilon}'(T_{N(j)}+\i(r-\sigma))\log
|S_\alpha(\tfrac{1}{2}+\sigma-r+\i T_{N(j)})|dr
\end{split}
\end{equation}
exists as $j\to\infty$. We have used evenness of $h_{\epsilon}$ and $h'_{\epsilon}(-\bar{\rho})=-h'_{\epsilon}(\bar{\rho}) =-\overline{h'_{\epsilon}(\rho)}$. 
From Proposition \ref{seqbound} we know that there exists a positive constant $c(\Gamma)$ such that 
\begin{equation}
\log \lbrace c(\Gamma)e^{-16C_{2}(\Gamma)T_{N(j)}^{2}\ln T_{N(j)}}|S_\alpha(\tfrac{1}{2}+\i\rho)|\rbrace\leq0
\end{equation} 
for all $j$ and $\rho\in[T_{N(j)},T_{N(j)}-\i\sigma]$. From Lemma 25 we have for $\rho\in[T_{N(j)}-\i\sigma,T_{N(j)}]$ and $T_{N(j)}$ large that $-\Re h_{\epsilon}'(\rho)=|\Re h_{\epsilon}'(\rho)|>\!\!>T_{N(j)}^{-2-\epsilon}$ uniformly in $j$ and $\rho$. Let $E_j=e^{-16C_{2}(\Gamma)T_{N(j)}^{2}\ln T_{N(j)}}$. Hence
\begin{equation}
\begin{split}
&T_{N(j)}^{-2-\epsilon}\int_{0}^{\sigma}\left|\log\left\{ c(\Gamma)E_j|S_\alpha(\tfrac{1}{2}+\sigma-r+\i T_{N(j)})|\right\}\right|dr\\
\ll\;&\int_{0}^{\sigma}|\Re h_{\epsilon}'(T_{N(j)}+\i(r-\sigma))|\\
&\;\times\left|\log\left\{ c(\Gamma)E_j|S_\alpha(\tfrac{1}{2}+\sigma-r+\i T_{N(j)})|\right\}\right|dr\\
=\;&\Bigg|\int_{0}^{\sigma}\Re h_{\epsilon}'(T_{N(j)}+\i(r-\sigma))\\
&\;\times\log\left\{ c(\Gamma)E_j|S_\alpha(\tfrac{1}{2}+\sigma-r+\i T_{N(j)})|\right\} dr\Bigg|
\end{split}
\end{equation}
which converges as $j\to\infty$. It follows for any $\delta>0$
\begin{equation}
\begin{split}
&\int_{T_{N(j)}-\i\sigma}^{T_{N(j)}}\big|\log|S_\alpha(\tfrac{1}{2}+\i\rho)|\big||d\rho|\\
=&\int_{0}^{\sigma}\big|\log\lbrace c(\Gamma)E_j|S_\alpha(\tfrac{1}{2}+\sigma-r+\i T_{N(j)})\rbrace\big|dr\\
&+O(T_{N(j)}^{2}\ln T_{N(j)})\\
\ll&\,T_{N(j)}^{2+\delta}.
\end{split}
\end{equation}
\end{proof}

We can now apply Proposition \ref{prop24} to derive the vanishing of the sequence of boundary terms \eqref{bdterms}.
\begin{thm}\label{thm25}
Let $\delta,\epsilon>0$ and $\lbrace T_{N(j)}\rbrace_{j}$ as above. Then for any $h\in H_{\sigma+\epsilon,\delta}$
\begin{equation}
\lim_{T_{N(j)}\to\infty}\int_{T_{N(j)}-\i\sigma}^{T_{N(j)}}h(\rho)\frac{S'_\alpha}{S_\alpha}(\tfrac{1}{2}+\i\rho)d\rho=0.
\end{equation}
\end{thm}
\begin{proof}
We follow the same lines as in Proposition \ref{prop24}. In exactly the same way as in the proof above we obtain the identity \eqref{rearr} for $h\in H_{\sigma+\epsilon,\delta}\subset H_{\sigma,\delta}$. So
\begin{equation}
\begin{split}
&\lim_{j\to\infty}\left\{\int_{T_{N(j)}-\i\sigma}^{T_{N(j)}}+\int_{-T_{N(j)}}^{-T_{N(j)}-\i\sigma}\right\}h(\rho)\frac{d}{d\rho}\log S_\alpha(\tfrac{1}{2}+\i\rho)d\rho\\
=&\lim_{j\to\infty}\int_{T_{N(j)}-\i\sigma}^{T_{N(j)}}\lbrace h'(\rho)-h'(-\bar{\rho})\rbrace\log|S_\alpha(\tfrac{1}{2}+\i\rho)|d\rho.
\end{split}
\end{equation}
The limit vanishes since
\begin{equation}
\begin{split}
&\int_{T_{N(j)}-\i\sigma}^{T_{N(j)}}|\lbrace h'(\rho)-h'(-\bar{\rho})\rbrace|\big|\log|S_\alpha(\tfrac{1}{2}+\i\rho)|\big||d\rho|\\
\ll\;&T_{N(j)}^{-2-\delta}\int_{T_{N(j)}-\i\sigma}^{T_{N(j)}}\big|\log|S_\alpha(\tfrac{1}{2}+\i\rho)|\big|d\rho|.
\end{split}
\end{equation}
where we have used Proposition \ref{prop24} and observe that by Cauchy's theorem $h\in H_{\sigma+\epsilon,\delta}$ implies
\begin{equation}
|h'(\rho)|\ll(1+|\Re\rho|)^{-2-\delta}
\end{equation}
uniformly in $|\Im\rho|\leq\sigma$.
\end{proof}

As an application of Theorem \ref{thm25} we can now prove the trace formula. Let $h\in H_{\sigma,\delta}$ for any $\sigma>\tfrac{1}{2}$ and $\delta>0$. We will make a specific choice of $\sigma$ below. We observe that
\begin{equation}
\begin{split}
&\int_{-\i\sigma-T}^{-\i\sigma+T}h(\rho)\frac{S'_\alpha}{S_\alpha}(\tfrac{1}{2}+\i\rho)d\rho\\
=\,&h(-\i\sigma+T)\log S_\alpha(\tfrac{1}{2}+\sigma+\i T)\\
&-h(-\i\sigma-T)\log S_\alpha(\tfrac{1}{2}+\sigma-\i T)\\
&-\int_{-\i\sigma-T}^{-\i\sigma+T}h'(\rho)\log S_\alpha(\tfrac{1}{2}+\i\rho)d\rho
\end{split}
\end{equation}
and, because of \eqref{log}, \eqref{Im} and the decay of $h$,
\begin{equation}
\lim_{T\to\infty}h(-\i\sigma\pm T)\log S_\alpha(\tfrac{1}{2}+\sigma\pm\i T)=0.
\end{equation}
Thus we can take $T\to\infty$ in \eqref{prop16} to obtain
\begin{equation}\label{pretrace}
\begin{split}
&\sum_{j\geq-M}h(\rho^{\alpha}_{j})-\sum_{j\geq-M}h(\rho_{j})\\
=\;&-\frac{1}{2\pi\i}\int_{-\i\sigma-\infty}^{-\i\sigma+\infty}h'(\rho)\log S_\alpha(\tfrac{1}{2}+\i\rho)d\rho\\
&+\tfrac{1}{2}\delta_{\Gamma}h(0)+\frac{1}{4\pi}\int_{-\infty}^{+\infty}h(\rho)\frac{\theta'_\alpha}{\theta_\alpha}(\tfrac{1}{2}+\i\rho)d\rho.
\end{split}
\end{equation}
where we have used Theorem \ref{thm25}, Proposition \ref{lim}, the existence of the limit \eqref{limbound}, the standard bound on the unperturbed cuspidal spectrum and our choice of generic $\alpha$ which ensures that the sum over perturbed eigenvalues is finite. Since $S_\alpha(s)$ depends continuously on $\alpha\neq0$ and Proposition \ref{lim} gives us control over the resonances we can extend \eqref{pretrace} to the countable set of non-generic values of $\alpha$. In particular Proposition \ref{lim} and the fact that the zeros of of $S_\alpha(s)$ move continuously under variation of $\alpha\neq0$ imply that the trace over unperturbed eigenvalues will converge absolutely.

The first term on the RHS can now be split as follows
\begin{equation}
\begin{split}
&-\frac{1}{2\pi\i}\int_{-\i\sigma-\infty}^{-\i\sigma+\infty}h'(\rho)\log S_\alpha(\tfrac{1}{2}+\i\rho)d\rho\\
=&\frac{1}{2\pi\i}\int_{-\i\sigma-\infty}^{-\i\sigma+\infty}h(\rho)\frac{m_\Gamma\beta\psi'(\tfrac{1}{2}+\i\rho)}{1+m_\Gamma\beta\psi(\tfrac{1}{2}+\i\rho)}d\rho\\
&+\frac{1}{2\pi\i}\int_{-\i\sigma-\infty}^{-\i\sigma+\infty}h'(\rho)\log\left\{1+\frac{G_{1/2+\i\rho}^{\Gamma\setminus\scrI}(z_0,z_0)}{1+m_\Gamma\beta\psi(\tfrac{1}{2}+\i\rho)}\right\}d\rho
\end{split}
\end{equation}
where
$$G_{1/2+\i\rho}^{\Gamma\setminus\scrI}(z_0,z_0)=\sum_{\gamma\in\Gamma\setminus\scrI}G_{1/2+\i\rho}(z_0,\gamma z_0).$$
The following bound justifies expanding the logarithm on the r.h.s. into a power series for sufficiently large $\sigma$.
\begin{lem}
Let $\sigma\geq\frac{1}{2}$ be large enough such that \eqref{condition} is satisfied. Then
\begin{equation}
\left|\frac{G_{1/2+\i\rho}^{\Gamma\setminus\scrI}(z_0,z_0)}{1+m_\Gamma\beta\psi(\tfrac{1}{2}+\i\rho)}\right|<1.
\end{equation}
\end{lem}
\begin{proof}
For $\Im\rho=-\sigma$ we have the bound (cf. Lemma 5 in \cite{U}, which works for any discrete subgroup)
\begin{equation}\label{bound on remainder}
|G_{1/2+\i\rho}^{\Gamma\setminus\scrI}(z_0,z_0)|\ll \frac{1}{\sqrt{\sigma}}.
\end{equation}
Also note $|\psi(\tfrac{1}{2}+\i\rho)|\gg\log(\tfrac{1}{2}+\sigma)$ for $\Im\rho=-\sigma$. So there exists a constant $C(\Gamma,\alpha,z_0)>0$ which depends on the implied constant in \eqref{bound on remainder}, such that \eqref{condition} implies the result. 
\end{proof}

Following the arguments of the compact case (cf. \cite{U}, eq. (5.47)-(5.58), pp. 17-19) we obtain
\begin{equation}
\begin{split}
&-\frac{1}{2\pi\i}\int_{-\i\sigma-\infty}^{-\i\sigma+\infty}h'(\rho)\log S_\alpha(\tfrac{1}{2}+\i\rho)d\rho\\
=&\frac{1}{2\pi}\int_{-\i\nu-\infty}^{-\i\nu+\infty}h(\rho)\frac{m_\Gamma\beta\psi'(\tfrac{1}{2}+\i\rho)}{1+m_\Gamma\beta\psi(\tfrac{1}{2}+\i\rho)}d\rho\\
&+\sum_{k=1}^{\infty}\,\beta^{k}\sum_{\gamma_{1},\cdots,\gamma_{k}\in\Gamma\setminus\scrI}
\int_{l_{\gamma_{1},z_{0}}}^{\infty}\cdots\int_{l_{\gamma_{k},z_{0}}}^{\infty}\frac{g_{\beta,k}(t_{1}+...+t_{k})\prod_{n=1}^{k}dt_{n}}{\prod_{n=1}^{k}\sqrt{\cosh t_{n}-\cosh l_{\gamma_{n},z_{0}}}}
\end{split}
\end{equation}
where $\nu=\nu(\alpha)$ is chosen as above and
\begin{equation}
g_{\beta,k}(t)=\frac{(-1)^{k}}{2\pi\i k}\int_{-\i\nu-\infty}^{-\i\nu+\infty}\frac{h'(\rho)e^{-\i\rho t}d\rho}{(1+m_{\Gamma}\beta\psi(\tfrac{1}{2}+\i\rho))^{k}}.
\end{equation}

Because of the identity
\begin{equation}
\varphi_\alpha(s)=\theta_\alpha(s)\varphi(s)
\end{equation}
we rewrite the second term as
\begin{equation}
\frac{1}{4\pi}\int_{-\infty}^{\infty}h(\rho)\left\{\frac{\varphi'_\alpha}{\varphi_\alpha}(\tfrac{1}{2}+\i\rho)-\frac{\varphi'}{\varphi}(\tfrac{1}{2}+\i\rho)\right\}d\rho.
\end{equation}
This concludes the proof of Theorem \ref{thm3}.

\bibliographystyle{plain}

\end{document}